\documentclass[11pt]{article}

\newlength{\actualtopmargin}
\newlength{\actualsidemargin}
\usepackage[tbtags]{amsmath}
\usepackage{amsthm}
\usepackage{amssymb}
\usepackage{amsfonts}
\usepackage{hyperref}

\theoremstyle{plain}
  \newtheorem{theorem}{Theorem}
  \newtheorem{lemma}[theorem]{Lemma}
  \newtheorem{corollary}[theorem]{Corollary}
  \newtheorem{proposition}[theorem]{Proposition}
  
\theoremstyle{definition}
  \newtheorem{definition}[theorem]{Definition}

\theoremstyle{remark}
  \newtheorem*{remark}{Remark}
\theoremstyle{plain}
  \newtheorem*{theorem*}{Theorem}
  \newtheorem*{lemma*}{Lemma}
  \newtheorem*{corollary*}{Corollary}
  \newtheorem*{proposition*}{Proposition}
  \newtheorem*{claim*}{Claim}

\newenvironment{step}
  {
    \begin{enumerate}

  }
  {\end{enumerate}}

\newenvironment{algorithm*}[1]
  {
    \begin{center}
      \hrulefill\\
      \textbf{#1}\\
  }
  {
    \vspace{-\baselineskip}
    \hrulefill
    \end{center}
  }

\newenvironment{protocol*}[1]
  {
    \begin{center}
      \hrulefill\\
      \textbf{#1}\\
  }
  {
    \vspace{-\baselineskip}
    \hrulefill
    \end{center}
  }

\newlength{\itemwidth}
\newlength{\descriptionwidth}
\newenvironment{promiseproblem*}[4]
  {
    \begin{center}
      \hrulefill\\
      \textbf{\textsc{#1}}\\
      \settowidth{\itemwidth}{\textbf{Yes Instances:}}
      \setlength{\descriptionwidth}{\textwidth}
      \addtolength{\descriptionwidth}{-\itemwidth}
      \addtolength{\descriptionwidth}{-\labelsep}
      \begin{description}
        \item[\parbox{\itemwidth}{Input:}]
          \parbox[t]{\descriptionwidth}{#2}
        \item[\parbox{\itemwidth}{Yes Instances:}]
          \parbox[t]{\descriptionwidth}{#3}
        \item[\parbox{\itemwidth}{No Instances:}]
          \parbox[t]{\descriptionwidth}{#4\\}
      \end{description}
  }
  {
    \vspace{-1.1\baselineskip}
    \hrulefill
    \end{center}
  }

\newcommand{\bbC}{\mathbb{C}}

\newcommand{\bbN}{\mathbb{N}}

\newcommand{\bbZ}{\mathbb{Z}}

\newcommand{\bfU}{\mathbf{U}}

\newcommand{\classfont}{\mathrm}

\newcommand{\spacefont}{\mathcal}
\newcommand{\registerfont}{\mathsf}

\newcommand{\spaceB}{\spacefont{B}}

\newcommand{\spaceH}{\spacefont{H}}

\newcommand{\regB}{\registerfont{B}}
\newcommand{\regC}{\registerfont{C}}

\newcommand{\regM}{\registerfont{M}}

\newcommand{\regQ}{\registerfont{Q}}
\newcommand{\regR}{\registerfont{R}}

\newcommand{\regV}{\registerfont{V}}

\newcommand{\Complex}{\bbC}

\newcommand{\Natural}{\bbN}
\newcommand{\Integers}{\bbZ}
\newcommand{\Nonnegative}{{\Integers^+}}

\newcommand{\Binary}{{\{ 0, 1 \}}}

\newcommand{\Unitary}{\bfU}

\newcommand{\SU}{\mathrm{SU}}

\newcommand{\unitaryQMASPACE}{{\classfont{QMA}_\Unitary\classfont{SPACE}}}

\newcommand{\NP}{\classfont{NP}}

\newcommand{\BQP}{\classfont{BQP}}

\newcommand{\PSPACE}{\classfont{PSPACE}}

\newcommand{\PrQPSPACE}{\classfont{PrQPSPACE}}

\newcommand{\unitaryQMAPSPACE}{{\classfont{QMA}_\Unitary\classfont{PSPACE}}}
\newcommand{\unitaryQPSPACE}{{\classfont{Q}_\Unitary\classfont{PSPACE}}}
\newcommand{\RevPSPACE}{\classfont{RevPSPACE}}

\newcommand{\unitaryBQL}{{\classfont{BQ}_\Unitary\classfont{L}}}

\newcommand{\unitaryQMAL}{{\classfont{QMA}_\Unitary\classfont{L}}}

\newcommand{\unitaryQL}{{\classfont{Q}_\Unitary\classfont{L}}}
\newcommand{\EXP}{\classfont{EXP}}

\newcommand{\QMA}{\classfont{QMA}}
\newcommand{\MA}{\classfont{MA}}

\newcommand{\MG}{\classfont{MG}}

\newcommand{\op}[1]{\operatorname{\mathnormal{#1}}}

\newcommand{\tensor}{\otimes}

\newcommand{\ket}[1]{\lvert #1 \rangle}

\newcommand{\bigket}[1]{\bigl\lvert #1 \bigr\rangle}

\newcommand{\ketbra}[1]{\lvert #1 \rangle \langle #1 \rvert}

\newcommand{\conjugate}[1]{{#1^\dagger}}

\newcommand{\tr}{\operatorname{tr}}

\newcommand{\abs}[1]{\lvert #1 \rvert}

\newcommand{\ceil}[1]{\lceil #1 \rceil}

\newcommand{\bigceil}[1]{\bigl\lceil #1 \bigr\rceil}
\newcommand{\Bigceil}[1]{\Bigl\lceil #1 \Bigr\rceil}
\newcommand{\biggceil}[1]{\biggl\lceil #1 \biggr\rceil}

\newcommand{\function}[3]{{#1 \colon #2 \to #3}}

\newcommand{\NOT}{\mathrm{NOT}}

\newcommand{\INCR}[1]{{\op{U}_{+1}(\Integers_{#1})}}

\newcommand{\textlog}{\mathrm{log}}

\newcommand{\acc}{\mathrm{acc}}

\newcommand{\init}{\mathrm{init}}

\newcommand{\yes}{\mathrm{yes}}
\newcommand{\no}{\mathrm{no}}

\begin{document}

\sloppy


\title{
  \Large{
    \textbf{
      Space-Efficient Error Reduction for Unitary Quantum Computations
    }
  }
}

\author{
  Bill Fefferman\footnotemark[1]\\
  \and
  Hirotada Kobayashi\footnotemark[2]\\
  \and
  Cedric Yen-Yu Lin\footnotemark[1]\\
  \and
  Tomoyuki Morimae\footnotemark[3]\\
  \and
  Harumichi Nishimura\footnotemark[4]
}

\date{}

\maketitle
\thispagestyle{empty}
\pagestyle{plain}
\setcounter{page}{0}

\renewcommand{\thefootnote}{\fnsymbol{footnote}}

\vspace{-5mm}

\begin{center}
  \large{
    \footnotemark[1]%
    Joint Center for Quantum Information and Computer Science\\
    University of Maryland\\
    College Park, MD, USA\\
    [2.5mm]
    \footnotemark[2]%
    Principles of Informatics Research Division\\
    National Institute of Informatics\\
    Tokyo, Japan\\
    [2.5mm]
    \footnotemark[3]%
    Advanced Scientific Research Leaders Development Unit\\
    Gunma University\\
    Kiryu, Gunma, Japan\\
    [2.5mm]
    \footnotemark[4]%
    Department of Computer Science and Mathematical Informatics\\
    Graduate School of Information Science\\
    Nagoya University\\
    Nagoya, Aichi, Japan
  }\\
  [5mm]
  \large{27 April 2016}\\
  [8mm]
\end{center}

\renewcommand{\thefootnote}{\arabic{footnote}}


\begin{abstract}
  This paper develops general space-efficient methods for error reduction for unitary quantum computation.
  Consider a polynomial-time quantum computation with completeness~$c$ and soundness~$s$,
  either with or without a witness
  (corresponding to $\QMA$ and $\BQP$, respectively).
  To convert this computation
  into a new computation with error at most $2^{-p}$,
  the most space-efficient method known requires extra workspace of 
  ${\op{O} \bigl( p \log \frac{1}{c-s} \bigr)}$~qubits.
  This space requirement is too large for scenarios like logarithmic-space quantum computations.
  This paper presents error-reduction methods
  for unitary quantum computations
  (i.e., computations without intermediate measurements)
  that require extra workspace of just ${\op{O} \bigl( \log \frac{p}{c-s} \bigr)}$~qubits.
  This in particular gives the first methods of strong amplification
  for logarithmic-space unitary quantum computations with two-sided bounded error.
  This also leads to a number of consequences in complexity theory,
  such as
  the uselessness of quantum witnesses in bounded-error logarithmic-space unitary quantum computations,
  the $\PSPACE$ upper bound for QMA with exponentially-small completeness-soundness gap,
  and strong amplification for matchgate computations.
\end{abstract}

\clearpage


\section{Introduction}
\label{Section: introduction}


\subsection{Background}
\label{Subsection: background}

A very basic topic in various models of quantum computation is
whether computation error can be efficiently reduced within a given model.
For polynomial-time bounded error quantum computation,
the most standard model of quantum computation,
the computation error can be made exponentially small 
via a simple repetition followed by a threshold-value decision.
This justifies the choice of $2/3$~and~$1/3$ for the completeness and soundness parameters
in the definition of the corresponding complexity class~$\BQP$.
This is also the case for quantum Merlin-Arthur (QMA) proof systems,
another central model of quantum computation
that models a quantum analogue of $\NP$ (more precisely, $\MA$),
and the resulting class~$\QMA$ may again be defined
with completeness and soundness parameters~$2/3$~and~$1/3$.

An undesirable feature of the simple repetition-based error reduction above
is that the necessary workspace enlarges linearly with respect to the number of repetitions.
More explicitly, for a given $p$,
the number of repetitions necessary to achieve an error of $2^{-p}$
is ${\op{O} \bigl( \frac{p}{(c-s)^2} \bigr)}$,
and thus both the workspace size and the witness size become
${\op{O} \bigl( \frac{p}{(c-s)^2} \bigr)}$~times larger.
This implies that
the simple repetition-based method is no longer useful
when either the workspace size or the witness size is required to be logarithmically bounded.

Marriott~and~Watrous~\cite{MarWat05CC} developed
a more sophisticated method of error reduction for QMA proof systems
that does not increase the witness size at all.
For a given $p$,
their method still requires ${\op{O} \bigl( \frac{p}{(c-s)^2} \bigr)}$~calls
of the original computation and its inverse
to achieve the computation error~$2^{-p}$,
but the method reuses both the workspace and the witness
every time it calls the original computation and its inverse.
Hence, the witness size never increases in their method.
This is a strong property that allows them to show
the uselessness of logarithmic-size quantum witnesses in QMA proof systems
(i.e., ${\QMA_\textlog = \BQP}$,
where $\QMA_\textlog$ is the class of problems having QMA proof systems
with logarithmic-size quantum witnesses).
Their method is also more efficient in workspace size than the simple repetition-based method,
but still requires extra workspace of size~${\op{O} \bigl( \frac{p}{(c-s)^2} \bigr)}$,
as it must record outcomes of all the calls of the original computation and its inverse.

Nagaj,~Wocjan,~and~Zhang~\cite{NagWocZha09QIC}
succeeded in reducing to ${\op{O} \bigl( \frac{p}{c-s} \bigr)}$
the number of calls of the original computation and its inverse
necessary to achieve the computation error~$2^{-p}$ for a given $p$,
while keeping the witness size unchanged.
Their method makes use of the phase-estimation algorithm,
an essential component of many quantum algorithms
including the celebrated factoring algorithm.
To achieve error~$2^{-p}$ for a given $p$,
their method must repeat ${\op{O}(p)}$~times the phase-estimation algorithm
with precision of at least ${\op{O} \bigl( \log \frac{1}{c-s} \bigr)}$~bits
and record all these estimated phases.
Hence, this phase-estimation-based method uses extra workspace of size~${\op{O} \bigl( p \log \frac{1}{c-s} \bigr)}$.

As can be seen from above,
both of the Marriott-Watrous method and the phase-estimation-based method
are still insufficient for the case
where the workspace size must be logarithmically bounded.
No efficient error-reduction method is known
that keeps the size of additionally necessary workspace logarithmically bounded.
This is not limited to the case of QMA proof systems,
and in fact almost no efficient error-reduction method is known
even in the case of logarithmic-space quantum computations,
and in the case of space-bounded quantum computations in general.
The study of general space-bounded quantum computations
was initiated by Watrous~\cite{Wat99JCSS} based on quantum Turing machines.
Several models of space-bounded quantum computations have been proposed and investigated since then
in the literature~\cite{Wat01JCSS, Wat03CC, Wat09ECSS, JozKraMiyWat10RSPA, MelWat12ToC, TaS13STOC},
some considering only logarithmic-space quantum computations
and others treating general cases.
It is not known whether any of these models are computationally equivalent.
It is also not known whether error reduction is possible for logarithmic-space quantum computation
defined according to any of these models,
except the only known affirmative answer shown by Watrous~\cite{Wat01JCSS}
on computation of one-sided bounded error performed by logarithmic-space quantum Turing machines.
As negative evidence in the case where computational resources are too limited,
computation error cannot be reduced below a certain constant
for one-way quantum finite state automata~\cite{AmbFre98FOCS}.


\subsection{Main Result and Its Consequences}
\label{Subsection: main result and its consequences}

This paper presents a general method of strong and space-efficient error reduction
for \emph{unitary} quantum computations.
In particular, the method is applicable to logarithmic-space unitary quantum computations
and logarithmic-space unitary QMA proof systems.
All the results in this paper are model-independent
and hold with any model of space-bounded quantum computations
as long as it performs \emph{unitary} quantum computations.
The unitary model is not the most general in that it does not allow any intermediate measurements
(notice that the standard technique of simulating intermediate measurements by unitary gates
requires unallowably many ancilla qubits in the case of space-bounded computations),
but is arguably one of the most reasonable models of space-bounded quantum computation.

Let $\Natural$ and $\Nonnegative$ denote
the sets of positive and nonnegative integers, respectively.
Let ${\unitaryQMASPACE[l_\regV, l_\regM](c, s)}$ denote the class of problems
having QMA proof systems with completeness~$c$ and soundness~$s$,
where the verifier performs a \emph{unitary} quantum computation
that has no time bound but is restricted to use ${\op{l_\regV}(n)}$~private qubits
and to receive a quantum witness of ${\op{l_\regM}(n)}$~qubits
on every input of length~$n$.
The main result of this paper is
the following strong and space-efficient error-reduction for such QMA-type computations.

\begin{theorem}
  For any functions~$\function{p, l_\regV, l_\regM}{\Nonnegative}{\Natural}$
  and for any functions~$\function{c,s}{\Nonnegative}{[0,1]}$
  satisfying ${c > s}$,
  there exists a function~$\function{\delta}{\Nonnegative}{\Natural}$
  that is logarithmic with respect to ${\frac{p}{c - s}}$
  such that
  \[
    \unitaryQMASPACE[l_\regV, l_\regM](c, s)
    \subseteq
    \unitaryQMASPACE[l_\regV + \delta, l_\regM](1 - 2^{-p}, 2^{-p}).
  \]
  \label{Theorem: main theorem}
\end{theorem}

This paper presents three different proofs of this main theorem,
all of which are based on reductions
that are in space logarithmic and also in time polynomial with respect to $\frac{p}{c-s}$.
As will be found in Section~\ref{Section: space-efficient amplification methods},
the theorem can be proved by remarkably simple arguments.
Nevertheless, the theorem is very powerful in that it fruitfully leads to many consequences
that substantially deepen the understanding
on the power of QMA proof systems and quantum computations in general,
both in the space-bounded scenario and in the usual polynomial-time scenario.
In what follows,
a function~$\function{f}{\Nonnegative}{\Natural}$ is \emph{polynomially bounded}
if $f$ is polynomial-time computable and ${\op{f}(n)}$ is in ${\op{O}(n^d)}$ for some constant~${d > 0}$,
and is \emph{logarithmically bounded}
if $f$ is logarithmic-space computable and ${\op{f}(n)}$ is in ${\op{O}(\log n)}$.


\paragraph{Strong amplification for unitary BQL}

The first consequence of Theorem~\ref{Theorem: main theorem}
is a remarkably strong error-reducibility
in logarithmic-space unitary quantum computations.
Let ${\unitaryQL(c,s)}$ denote the class of problems
solvable by logarithmic-space unitary quantum computations with completeness~$c$ and soundness~$s$.
The following amplifiability is immediate from Theorem~\ref{Theorem: main theorem}
by taking a function~$p$ to be logarithmic-space computable and polynomially bounded,
functions~$c$~and~$s$ to be logarithmic-space computable and to satisfy ${c - s \geq 1/q}$
for some polynomially bounded function~$\function{q}{\Nonnegative}{\Natural}$,
a function~$l_\regV$ to be logarithmically bounded,
and a function~${l_\regM = 0}$.

\begin{corollary}
  For any polynomially bounded function~$\function{p}{\Nonnegative}{\Natural}$
  that is logarithmic-space computable
  and for any logarithmic-space computable functions~$\function{c,s}{\Nonnegative}{[0,1]}$
  satisfying ${c - s \geq 1/q}$
  for some polynomially bounded function~$\function{q}{\Nonnegative}{\Natural}$,
  \[
    \unitaryQL(c, s) \subseteq \unitaryQL(1 - 2^{-p}, 2^{-p}).
  \]
  \label{Corollary: amplification for unitary BQL}
\end{corollary}

This in particular justifies
defining the bounded-error class~$\unitaryBQL$ of logarithmic-space unitary quantum computations
by ${\unitaryBQL = \unitaryQL(2/3, 1/3)}$,
employing $2/3$ and $1/3$ for completeness and soundness parameters.
Before this work, Watrous~\cite{Wat01JCSS} showed a similar strong error-reducibility
in the case of one-sided bounded error,
and Corollary~\ref{Corollary: amplification for unitary BQL} extends this
to the two-sided bounded error case.


\paragraph{Uselessness of quantum witnesses in logarithmic-space unitary QMA}

Let ${\unitaryQMAL(c, s)}$ denote the class of problems
having logarithmic-space unitary QMA proof systems
(i.e., such systems in which a verifier performs a logarithmic-space unitary computation
upon receiving a logarithmic-size quantum witness)
with completeness~$c$ and soundness~$s$.
Similarly to Corollary~\ref{Corollary: amplification for unitary BQL},
the following amplifiability is immediate from Theorem~\ref{Theorem: main theorem}
by taking a function~$p$ to be logarithmic-space computable and polynomially bounded,
functions~$c$~and~$s$ to be logarithmic-space computable and to satisfy ${c - s \geq 1/q}$
for some polynomially bounded function~$\function{q}{\Nonnegative}{\Natural}$,
and functions~$l_\regV$~and~$l_\regM$ to be logarithmically bounded.

\begin{corollary}
  For any polynomially bounded function~$\function{p}{\Nonnegative}{\Natural}$
  that is logarithmic-space computable
  and for any logarithmic-space computable functions~$\function{c,s}{\Nonnegative}{[0,1]}$
  satisfying ${c - s \geq 1/q}$
  for some polynomially bounded function~$\function{q}{\Nonnegative}{\Natural}$,
  \[
    \unitaryQMAL(c, s) \subseteq \unitaryQMAL(1 - 2^{-p}, 2^{-p}).
  \]
  \label{Corollary: amplification for unitary QMAL}
\end{corollary}

Again this justifies
defining the bounded-error class~$\unitaryQMAL$ of logarithmic-space unitary QMA proof systems
by ${\unitaryQMAL = \unitaryQMAL(2/3, 1/3)}$.
By a standard technique of replacing a quantum witness
by a totally mixed state as a self-prepared witness
(to do this in a unitary computation, one can simply prepare sufficiently many
EPR pairs and then take a qubit from each pair),
Corollary~\ref{Corollary: amplification for unitary QMAL}
together with Corollary~\ref{Corollary: amplification for unitary BQL}
further implies the equivalence of $\unitaryQMAL$ and $\unitaryBQL$.

\begin{corollary}
  ${\unitaryQMAL = \unitaryBQL}$.
  \label{Corollary: unitary QMAL = unitary BQL}
\end{corollary}

As mentioned before,
Marriott~and~Watrous~\cite{MarWat05CC} showed the equivalence~${\QMA_\textlog = \BQP}$,
the uselessness of quantum witnesses of logarithmic size
in the standard QMA proof systems with a polynomial-time verifier.
In this respect,
Corollary~\ref{Corollary: unitary QMAL = unitary BQL}
states that quantum witnesses of logarithmic size
do not increase the power of logarithmic-space unitary quantum computations at all,
and indeed extends the result of Marriott and Watrous to logarithmic-space case.


\paragraph{Space-efficient amplification for QMA}

Let ${\QMA[l_\regV, l_\regM](c, s)}$ be the time-efficient version of
${\unitaryQMASPACE[l_\regV, l_\regM](c, s)}$,
i.e.,
the class of problems
having standard polynomial-time QMA proof systems with completeness~$c$ and soundness~$s$
in which a polynomial-time unitary quantum verifier
receives a quantum witness of ${\op{l_\regM}(n)}$~qubits
and uses workspace of ${\op{l_\regV}(n)}$~qubits
on every input of length~$n$.
As the reduction is in time polynomial with respect to $\frac{p}{c-s}$
in the proof of Theorem~\ref{Theorem: main theorem},
the following amplifiability is immediate from Theorem~\ref{Theorem: main theorem}
by taking functions~$p$,~$l_\regV$,~and~$l_\regM$ to be polynomially bounded,
and functions~$c$~and~$s$ to be polynomial-time computable and to satisfy ${c - s \geq 1/q}$
for some polynomially bounded function~$\function{q}{\Nonnegative}{\Natural}$.

\begin{corollary}
  For any polynomially bounded functions~$\function{p, l_\regV, l_\regM}{\Nonnegative}{\Natural}$
  and for any polynomial-time computable functions~$\function{c,s}{\Nonnegative}{[0,1]}$
  satisfying ${c - s \geq 1/q}$
  for some polynomially bounded function~$\function{q}{\Nonnegative}{\Natural}$,
  there exists a function~$\function{\delta}{\Nonnegative}{\Natural}$
  that is logarithmic with respect to ${\frac{p}{c - s}}$
  such that
  \[
    \QMA[l_\regV, l_\regM](c, s)
    \subseteq
    \QMA[l_\regV + \delta, l_\regM](1 - 2^{-p}, 2^{-p}).
  \]
  \label{Corollary: space-efficient amplification for QMA}
\end{corollary}

Recall that
the Marriott-Watrous amplification~\cite{MarWat05CC}
requires $\delta$ to be in ${\op{O} \bigl( \frac{p}{(c-s)^2} \bigr)}$
and the phase-estimation-based method by Nagaj,~Wocjan,~and~Zhang~\cite{NagWocZha09QIC}
requires $\delta$ to be in ${\op{O} \bigl( p \log \frac{1}{c-s} \bigr)}$,
instead of $\delta$ in ${\op{O} \bigl( \log \frac{p}{c-s} \bigr)}$
of Corollary~\ref{Corollary: space-efficient amplification for QMA}.
Hence, the methods in this paper are most space-efficient
among known error-reduction methods for standard QMA proof systems,
and also among those for $\BQP$.


\paragraph{Strong amplification for unitary QMAPSPACE}

Let ${\unitaryQPSPACE(c, s)}$ denote the class of problems
solvable by polynomial-space unitary quantum computations with completeness~$c$ and soundness~$s$,
and let ${\unitaryQMAPSPACE(c, s)}$ denote the class of problems
having polynomial-space unitary QMA proof systems
(i.e., such systems in which a verifier performs a polynomial-space unitary computation
upon receiving a polynomial-size quantum witness)
with completeness~$c$ and soundness~$s$.
The following corollary states the scaled-up versions of
Corollaries~\ref{Corollary: amplification for unitary BQL}~and~\ref{Corollary: amplification for unitary QMAL},
and again is immediate from Theorem~\ref{Theorem: main theorem}
by taking a function~$p$ to be polynomial-space computable and exponentially bounded,
functions~$c$~and~$s$ to be polynomial-space computable and to satisfy ${c - s \geq 2^{-q}}$
for some polynomially bounded function~$\function{q}{\Nonnegative}{\Natural}$,
and functions~$l_\regV$~and~$l_\regM$ to be polynomially bounded
(or a function~${l_\regM = 0}$ in the case of ${\unitaryQPSPACE(c, s)}$).

\begin{corollary}
  For any polynomially bounded function~$\function{p}{\Nonnegative}{\Natural}$
  and for any polynomial-space computable functions~$\function{c,s}{\Nonnegative}{[0,1]}$
  satisfying ${c - s \geq 2^{-q}}$
  for some polynomially bounded function~$\function{q}{\Nonnegative}{\Natural}$,
  the following two properties hold:
  \begin{itemize}
  \item[(i)]
    ${\unitaryQPSPACE(c, s) \subseteq \unitaryQPSPACE \bigl( 1 - 2^{- 2^p}, 2^{- 2^p} \bigr)}$.
  \item[(ii)]
    ${\unitaryQMAPSPACE(c, s) \subseteq \unitaryQMAPSPACE \bigl( 1 - 2^{- 2^p}, 2^{- 2^p} \bigr)}$.
  \end{itemize}
  \label{Corollary: amplification for unitary QPSPACE and QMAPSPACE}
\end{corollary}

Again by a standard technique of replacing a quantum witness
by a totally mixed state as a self-prepared witness,
the following corollary follows from
Corollary~\ref{Corollary: amplification for unitary QPSPACE and QMAPSPACE}
together with the fact that ${\RevPSPACE = \PrQPSPACE = \PSPACE}$~\cite{Ben89SIComp, Wat99JCSS},
where $\RevPSPACE$ and $\PrQPSPACE$ are the complexity classes corresponding to
deterministic polynomial-space reversible computations
and unbounded-error polynomial-space quantum computations,
respectively.

\begin{corollary}
  For any polynomial-space computable functions~$\function{c,s}{\Nonnegative}{[0,1]}$
  satisfying ${c - s \geq 2^{-q}}$
  for some polynomially bounded function~$\function{q}{\Nonnegative}{\Natural}$,
  \[
    \unitaryQMAPSPACE(c, s) = \PSPACE.
  \]
  \label{Theorem: unitary QMAPSPACE with exp-small gap is in PSPACE}
\end{corollary}

Now the $\PSPACE$ upper bound immediately follows
for the class of problems having standard polynomial-time QMA proof systems
with exponentially small completeness-soundness gap.
More precisely,
for the class~${\QMA(c,s)}$ of problems having standard polynomial-time QMA proof systems
with completeness~$c$ and soundness~$s$,
the following corollary holds.

\begin{corollary}
  For any polynomially bounded function~$\function{p}{\Nonnegative}{\Natural}$
  and for any polynomial-time computable functions~$\function{c,s}{\Nonnegative}{[0,1]}$
  satisfying ${c - s \geq 2^{-q}}$
  for some polynomially bounded function~$\function{q}{\Nonnegative}{\Natural}$,
  \[
    \QMA(c, s) \subseteq \PSPACE.
  \]
  \label{Corollary: QMA with exp-small gap is in PSPACE}
\end{corollary}

For QMA proof systems
with exponentially small completeness-soundness gap,
the $\PSPACE$ upper bound was known previously only for the one-sided-error case
(following from the result in Ref.~\cite{ItoKobWat12ITCS}),
and only the $\EXP$ upper bound was known for the two-sided-error case
(following from the result in Ref.~\cite{KitWat00STOC}).
Natarajan~and~Wu~\cite{NatWu16Private}
independently proved
a statement equivalent to Corollary~\ref{Corollary: QMA with exp-small gap is in PSPACE}.
In fact, statements equivalent to Corollary~\ref{Corollary: QMA with exp-small gap is in PSPACE}
were also proved with different proofs
independently by the first and third authors of the present paper in Ref.~\cite{FefLin16arXiv-1}
(see Ref.~\cite{FefLin16arXiv-2} also)
and by the complement subset of the present authors.
The first and third authors of the present paper
further proved in Refs.~\cite{FefLin16arXiv-1, FefLin16arXiv-2}
that the converse of Corollary~\ref{Corollary: QMA with exp-small gap is in PSPACE} also holds,
i.e.,
$\PSPACE$ is characterized by
QMA proof systems with exponentially small completeness-soundness gap.


\paragraph{Strong amplification for matchgate computations}

A matchgate is defined to be a two-qubit gate of the form~${G(A,B)}$
corresponding to the four-by-four unitary matrix
in which the four corner elements form $A$
and the four inner-square elements form $B$
for matrices~$A$~and~$B$ in ${\SU(2)}$,
and all the other elements are $0$.
A matchgate circuit is a quantum circuit such that:
(i) the input state is a computational basis state,
(ii) all the gates of the circuit are matchgates which are applied to two neighbor qubits,
and (iii) the output is a final measurement in the computational basis on any single qubit.
Matchgate computations were introduced and proved classically simulable by Valiant~\cite{Val02SIComp}.
Terhal~and~DiVincenzo~\cite{TerDiV02PRA} related them to noninteracting-fermion quantum circuits.
Let ${\MG(c, s)}$ denote the class of problems solvable by polynomial-time matchgate computations
with completeness~$c$ and soundness~$s$.
Using the equivalence of polynomial-time matchgate computations and logarithmic-space unitary computations
shown by Jozsa,~Kraus,~Miyake,~and~Watrous~\cite[Corollary 3.3]{JozKraMiyWat10RSPA},
the following is immediate from Corollary~\ref{Corollary: amplification for unitary BQL}.

\begin{corollary}
  For any polynomially bounded function~$\function{p}{\Nonnegative}{\Natural}$
  that is logarithmic-space computable
  and for any logarithmic-space computable functions~$\function{c,s}{\Nonnegative}{[0,1]}$
  satisfying ${c - s \geq 1/q}$
  for some polynomially bounded function~$\function{q}{\Nonnegative}{\Natural}$,
  \[
    \MG(c, s) \subseteq \MG(1 - 2^{-p}, 2^{-p}).
  \]
  \label{Corollary: amplification for matchgate computations}
\end{corollary}


\subsection{Roadmap}
\label{Subsection: organization}

We assume familiarity with basic quantum formalism
(see Refs.~\cite{NieChu00Book, KitSheVya02Book, Wil13Book}, for instance).

Section~\ref{Section: overview}
provides outlines of three different proofs of the main theorem.
Subsection~\ref{Subsection: overview of simple construction based on phase estimation}
overviews the simplest construction among the three,
which is based on phase estimation.
Subsection~\ref{Subsection: overview of hybrid construction of phase estimation and Marriott-Watrous}
then briefly explains a hybrid construction based on both phase estimation and the Marriott-Watrous amplification,
which is most efficient among the three
in terms of the number of calls of the original unitary transformation of the verifier.
Subsection~\ref{Subsection: overview of exactly implementable construction based on a random guess}
sketches an alternative construction based on random guess,
which is exactly implementable
when the Hadamard and any classical reversible transformations are exactly implementable.
Section~\ref{Section: space-bounded unitary QMA}
presents precise definitions of the model of space-bounded unitary quantum Merlin-Arthur proof systems
and associated complexity classes.
Section~\ref{Section: basic procedures} describes several procedures
that are used in the main error-reduction procedures of this paper.
Finally, Section~\ref{Section: space-efficient amplification methods} 
provides the three proofs of the main theorem rigorously.


\section{Overview of Proofs}
\label{Section: overview}

This section provides outlines of the three different proofs of the main theorem.
Consider any unitary transformation~$V_x$ of the verifier on input $x$,
and let $p_\acc$ be the maximum acceptance probability of it
(and thus, ${p_\acc \geq \op{c}(\abs{x})}$ for yes instances,
and ${p_\acc \leq \op{s}(\abs{x})}$ for no instances).


\subsection{Simple Construction Based on Phase Estimation}
\label{Subsection: overview of simple construction based on phase estimation}

The first construction of space-efficient amplification
is very simple and mainly based on phase estimation.
The key idea is to first use phase estimation
so that it just reduces computation error \emph{mildly} to be polynomially small
rather than directly to be exponentially small.
The point is that the phase estimation is performed only once rather than multiple times.
By essentially taking the AND of the polynomially many attempts of this mildly amplified procedure,
one then achieves exponentially small soundness
with keeping sufficiently large completeness (say, ${1/2}$).
Finally, one makes completeness exponentially close to one
while keeping exponentially small soundness,
which is done by essentially taking the OR of the polynomially many attempts of the procedure constructed so far.

More precisely,
let $\spaceH$ be the Hilbert space over which $V_x$ acts,
and let $I_\spaceH$ be the identity operator over $\spaceH$.
Further let $\Pi_\init$ be the projection onto the subspace spanned by
the legal initial states of the QMA-type computation induced by $V_x$,
and let $\Pi_\acc$ be the projection onto the subspace spanned by
the accepting states of the QMA-type computation associated with $V_x$.
Consider the unitary operator~%
${Q_x = \bigl( 2 \conjugate{V_x} \Pi_\acc V_x - I_\spaceH \bigr) \bigl( 2 \Pi_\init - I_\spaceH \bigr)}$
corresponding to one iteration of the Grover-type algorithm induced by $V_x$.
First, one performs one-shot phase estimation associated with $Q_x$
with ${\op{l}(\abs{x})}$-bit precision
for a function~$\function{l}{\Nonnegative}{\Natural}$ defined by
${l = \bigceil{\log \frac{2 \pi}{\arccos \sqrt{s} - \arccos \sqrt{c}}}}$
and with \emph{mild} failure probability~$\frac{1}{\op{q}_1(\abs{x})}$,
where $q_1$ is a function in ${\op{O}(p)}$
(precisely speaking, ${q_1 = 2 (p + \ceil{\log (p+2)}) + 4}$).
From the property of the standard phase-estimation algorithm,
the number of additional qubits used by the resulting procedure
is determined by the function~${l + \bigceil{\log (\frac{q_1}{2} + 2)}}$,
which is at most linear in ${\log \frac{p}{c - s}}$
(in fact, at most ${\log \frac{p}{c - s}}$ plus a constant).
The acceptance probability is \emph{mildly} amplified to at least~${1 - \frac{1}{\op{q}_1(\abs{x})}}$
in the yes-instance case,
while it is \emph{mildly} reduced to at most~$\frac{1}{\op{q}_1(\abs{x})}$
in the no-instance case.

Let $V_x^{(1)}$ be the unitary operator corresponding to the procedure constructed so far.
Now repeat the following procedure ${\op{N}_1(\abs{x})}$~times
for ${N_1 = \bigceil{\frac{q_2}{2 \log q_1}}}$,
where $q_2$ is also a function in ${\op{O}(p)}$
(precisely speaking, ${q_2 = p + \ceil{\log (p+2)}}$ so that ${q_1 = 2q_2 + 4}$):
One applies $V_x^{(1)}$, and then increments a counter by $1$
if the state corresponds to a rejection state of it.
One further
applies $\conjugate{\bigl( V_x^{(1)} \bigr)}$,
the inverse of $V_x^{(1)}$,
and then increments a counter by one
if any of the work qubits of $V_x^{(1)}$ is in state~$\ket{1}$.
After the repetition, one accepts if and only if the counter value remains zero.
Intuitively, these repetitions try to take
the AND of the ${\op{N}_1(\abs{x})}$~attempts of $V_x^{(1)}$
(with some suitable initialization try by $\conjugate{\bigl( V_x^{(1)} \bigr)}$).
The rigorous analysis shows that
the initialization steps also contribute to taking AND,
so that this process is exactly equivalent to
taking the AND of ${2 \op{N}_1(\abs{x})}$~attempts of $V_x^{(1)}$.
The number of additional qubits used by the resulting procedure
is ${\op{O}(\log N_1)}$,
which is clearly at most linear in ${\log \frac{p}{c - s}}$.
The acceptance probability is thus reduced to at most~%
${\bigl( \frac{1}{\op{q}_1(\abs{x})} \bigr)^{2 \op{N}_1(\abs{x})} \leq 2^{- \op{q}_2(\abs{x})}}$
in the no-instance case,
while it is still at least~%
${1 - \frac{2 \op{N}_1(\abs{x})}{\op{q}_1(\abs{x})} > \frac{1}{2}}$
in the yes-instance case.

Let $V_x^{(2)}$ be the unitary operator corresponding to the procedure constructed so far.
Finally, one tries to take the OR of ${2 \op{N}_2(\abs{x})}$~attempts of $V_x^{(2)}$
for a function~$\function{N_2}{\Nonnegative}{\Natural}$ defined by
${N_2 = \ceil{\frac{p}{2}}}$,
which is done by performing a repetition similar to above.
The number of additional qubits used by the resulting procedure
is ${\op{O}(\log N_2)}$,
which is clearly at most linear in ${\log \frac{p}{c - s}}$.
The acceptance probability is amplified to at least~${1 - 2^{- \op{p}(\abs{x})}}$
in the yes-instance case,
while it is still at most~${2 \op{N}_2(\abs{x}) \cdot 2^{- \op{q}_2(\abs{x})} < 2^{- \op{p}(\abs{x})}}$
in the no-instance case,
as desired.


\subsection{Hybrid Construction of Phase Estimation and Marriott-Watrous}
\label{Subsection: overview of hybrid construction of phase estimation and Marriott-Watrous}

Recall that
the necessary number of calls of the (controlled) unitary transformation~$U$ is
${2^l \cdot \bigceil{\frac{1}{2 \varepsilon} + 2} - 1}$
for a phase estimation associated with $U$
precise to $l$~bits with failure probability~$\varepsilon$~\cite{NieChu00Book}.
Hence, a straightforward calculation shows that
the simple construction in the last subsection requires
${\op{O} \bigl( \frac{1}{c - s} \cdot \frac{p^3}{\log p} \bigr)}$~calls
of $V_x$ and its inverse.
This subsection presents an idea to construct a more efficient method
that uses ${\op{O} \bigl( \frac{1}{c - s} \cdot \frac{p^2}{\log p} \bigr)}$~calls
of $V_x$ and its inverse.
The idea here is to use phase estimation
so that it just achieves a \emph{very mild} computation error of some constant,
rather than polynomially small.
One then achieves polynomially small error by the Marriott-Watrous amplification.
The rest of the construction is essentially the same as in the simple construction in the last subsection.

More precisely, the construction first performs one-shot phase estimation
with ${\op{l}(\abs{x})}$-bit precision
for a function~$\function{l}{\Nonnegative}{\Natural}$ defined by
${l = \bigceil{\log \frac{2 \pi}{\arccos \sqrt{s} - \arccos \sqrt{c}}}}$
and with \emph{very mild} failure probability~$\frac{1}{4}$.
From the property of the standard phase-estimation algorithm,
the number of additional qubits used by the resulting procedure
is determined by the function~${l + 2}$,
which is at most ${\log \frac{1}{c - s}}$ plus a constant,
and thus, clearly at most linear in ${\log \frac{p}{c - s}}$
when the final targeted computation error is at most $2^{-p}$
for a function~$\function{p}{\Nonnegative}{\Natural}$.
The acceptance probability is \emph{very mildly} amplified to at least~$\frac{3}{4}$
in the yes-instance case,
while it is \emph{very mildly} reduced to at most~$\frac{1}{4}$
in the no-instance case.

Let $V_x^{(1)}$ be the unitary operator corresponding to the procedure constructed so far.
Next, one further reduces computation error still mildly to be polynomially small
by performing the Marriott-Watrous amplification.
By using ${\op{N}_1(\abs{x})}$~calls of $V_x^{(1)}$ and its inverse
for a function~$\function{N_1}{\Nonnegative}{\Natural}$ defined by
${N_1 = \bigceil{\frac{8 \log (2p)}{\log e}}}$,
the acceptance probability is \emph{mildly} amplified to at least~%
${1 - \frac{1}{4 \left( \op{p}(\abs{x}) \right)^2}}$
in the yes-instance case,
while it is \emph{mildly} reduced to at most~%
$\frac{1}{4 \left( \op{p}(\abs{x}) \right)^2}$
in the no-instance case.
The number of additional qubits used by the resulting procedure
is determined by the function~%
${2N_1 + \ceil{\log (2N_1 + 1)} + 1}$,
which is clearly at most linear in ${\log p}$
(and thus, at most linear in ${\log \frac{p}{c - s}}$ also).

Let $V_x^{(2)}$ be the unitary operator corresponding to the procedure constructed so far.
The rest of the construction is essentially the same as in the last subsection.
One can essentially take the AND of ${2 \op{N}_2(\abs{x})}$~attempts of $V_x^{(2)}$
for a function~$\function{N_2}{\Nonnegative}{\Natural}$ defined by
${N_2 = \bigceil{\frac{p}{2 \log (2p)}}}$
to achieve acceptance probability at least~${1 - \frac{1}{\op{p}(\abs{x})}}$ for yes instances
and at most~$2^{- 2 \op{p}(\abs{x})}$ for no instances.
Let $V_x^{(3)}$ be the resulting unitary operator.
One then essentially takes the OR of ${2 \op{N}_3(\abs{x})}$~attempts of $V_x^{(3)}$
for a function~$\function{N_3}{\Nonnegative}{\Natural}$ defined by
${N_3 = \ceil{\frac{p}{2 \log p}}}$
to achieve acceptance probability at least~${1 - 2^{- \op{p}(\abs{x})}}$ for yes instances
and at most~$2^{- \op{p}(\abs{x})}$ for no instances.

The total number of additional qubits required is clearly determined
by a function at most linear in ${\log \frac{p}{c - s}}$.
A straightforward calculation shows that
this construction uses
${\op{O} \bigl( \frac{1}{c - s} \cdot \frac{p^2}{\log p} \bigr)}$~calls
of $V_x$ and its inverse,
as claimed.


\subsection{Exactly Implementable Construction Based on a Random Guess}
\label{Subsection: overview of exactly implementable construction based on a random guess}

One small drawback of the previous two constructions is
that they are not exactly implementable
when implemented by quantum circuits with any gate set of finite size,
due to the use of the phase-estimation algorithm.
This subsection outlines an alternative construction
that is exactly implementable
when the Hadamard and any classical reversible transformations are exactly implementable.
The construction uses
${
  \op{O} \Bigl(
           \frac{1}{(c - s)^3} \cdot p^{\frac{5}{2}}
           +
           \frac{1}{(c - s)^3} \bigl( \log \frac{1}{c - s} \bigr)^{\frac{3}{2}} \cdot p
         \Bigr)
}$~calls
of $V_x$ and its inverse,
which is not so good as the second construction in Subsection~\ref{Subsection: overview of hybrid construction of phase estimation and Marriott-Watrous},
but is at least incomparable with the simple construction in Subsection~\ref{Subsection: overview of simple construction based on phase estimation}.

The idea is to guess $p_\acc$ with \emph{mild} precision of ${\op{l}(\abs{x})}$~bits,
where $\function{l}{\Nonnegative}{\Natural}$ is the function defined by
${l = \bigceil{\frac{1}{2} \log \frac{6q}{(c-s)^2}}}$
for a function~$\function{q}{\Nonnegative}{\Natural}$
defined by ${q = \bigceil{2 \bigl( p + \log \frac{6p}{c - s} + 1 \bigr)}}$
when the final targeted computation error is at most~$2^{-p}$
for a function~$\function{p}{\Nonnegative}{\Natural}$.
%
For each $j$ in ${\{1 ,\dotsc, 2^{\op{l}(\abs{x})}\}}$,
let ${r_j = j \cdot 2^{- \op{l}(\abs{x})}}$ be a possible guess of $p_\acc$.
Pick an integer~$k$ from ${\{1 ,\dotsc, 2^{\op{l}(\abs{x})}\}}$ uniformly at random,
and reject immediately if ${r_k = k \cdot 2^{- \op{l}(\abs{x})} < \op{c}(\abs{x})}$
(so that no $k$ can result in a good guess at $p_\acc$ for no instances).
Otherwise $r_k$ is used as a guess at $p_\acc$.
The point is that, for yes instances, there exists a choice of $k$ such that
${
  \abs{r_k - p_\acc}
  <
  2^{- \op{l}(\abs{x})}
  \leq
  \sqrt{\frac{(\op{c}(\abs{x}) - \op{s}(\abs{x}))^2}{6 \op{p}(\abs{x})}}
}$,
while for no instances,
it holds that
${\abs{r_k - p_\acc} > \op{c}(\abs{x}) - \op{s}(\abs{x})}$
for any choice of $k$.
Hence, by first applying the additive adjustment of acceptance probability~\cite{JorKobNagNis12QIC}
to obtain the unitary transformation~$V_{x,k}^{(1)}$ from $V_x$,
and then performing \textsc{Reflection Procedure}~\cite{KobLeGNis15SIComp} using $V_{x,k}^{(1)}$,
the acceptance probability can be \emph{mildly} amplified to at least
${1 - \frac{\left( \op{c}(\abs{x}) - \op{s}(\abs{x}) \right)^2}{6 \op{q}(\abs{x})}}$
in the yes-instance case,
if the chosen~$k$ corresponds to the appropriate guess~$r_k$,
while the acceptance probability is at most
${1 - \bigl( \op{c}(\abs{x}) - \op{s}(\abs{x}) \bigr)^2}$
for any guess~$r_k$.

Fix an index~$k$ of the guess~$r_k$
and let $V_{x,k}^{(2)}$ be the unitary operator corresponding to the procedure constructed so far.
As in the previous subsections,
one tries to essentially take the AND of ${2 \op{N}_2(\abs{x})}$~attempts of $V_{x,k}^{(2)}$
for a function~$\function{N_2}{\Nonnegative}{\Natural}$ defined by
${N_2 = \bigceil{\frac{q}{2 (c-s)^2}}}$.
The acceptance probability is
still at least~$\frac{1}{2}$ in the yes-instance case
when the appropriate guess~$r_k$ at $p_\acc$ is made,
while it is at most ${e^{- \op{q}(\abs{x})} < 2^{- \op{q}(\abs{x})}}$
for any guess~$r_k$ in the no-instance case.

Let $V_{x,k}^{(3)}$ be the unitary operator corresponding to the procedure constructed so far,
when the index~$k$ of $r_k$ is chosen.
Taking into account that $k$ is chosen uniformly at random,
the above argument results in a unitary transformation~$V_x^{(4)}$
that has acceptance probability at least
${
  2^{- \op{l}(\abs{x})} \cdot \frac{1}{2}
  >
  \frac{1}{4} \sqrt{\frac{(\op{c}(\abs{x}) - \op{s}(\abs{x}))^2}{6 \op{q}(\abs{x})}}
}$
in the yes-instance case
and at most~%
${
  2^{- \op{q}(\abs{x})}
  \leq
  2^{- \frac{\op{q}(\abs{x})}{2}}
  \cdot
  \bigl( \frac{\op{c}(\abs{x}) - \op{s}(\abs{x})}{12 \op{p}(\abs{x})} \bigr)
  \cdot
  2^{- \op{p}(\abs{x})}
}$
in the no-instance case.

Finally, as in the previous subsections,
one tries to essentially take the OR of ${2 \op{N}_4(\abs{x})}$~attempts of $V_x^{(4)}$
for a function~$\function{N_4}{\Nonnegative}{\Natural}$ defined by
${N_4 = \Bigceil{2 \sqrt{\frac{6q}{(c-s)^2}} \cdot p}}$.
The acceptance probability is amplified to
at least~${1 - 2^{- \op{p}(\abs{x})}}$ in the yes-instance case,
and is at most ${2^{- \op{p}(\abs{x})}}$
for any guess~$r_k$ in the no-instance case.


\section{Space-Bounded Unitary Quantum Merlin-Arthur Proof Systems}
\label{Section: space-bounded unitary QMA}

First we summarize some notations that are used in this paper.
Let ${\Sigma = \Binary}$ denote the binary alphabet set.
In this paper, all Hilbert spaces are complex and of dimension a power of two.
For a Hilbert space~$\spaceH$,
let $I_\spaceH$ denote the identity operator over $\spaceH$.
A quantum register is a set of single or multiple qubits.
For a quantum register~$\regR$,
let $I_\regR$ denote the identity operator over the Hilbert space associated with $\regR$.

A \emph{space-bounded unitary quantum Merlin-Arthur (QMA) proof system},
or simply called a \emph{QMA-type computation} throughout this paper,
is a space-bounded unitary quantum computation performed by a \emph{quantum verifier}~$V$.
As in the standard QMA proof system,
$V$ prepares a quantum register~$\regV$ corresponding to his/her private space,
all the qubits of which are initially in state~$\ket{0}$,
and receives a quantum register~$\regM$ storing an arbitrarily prepared quantum witness.
One of the qubit in $\regV$ is designated as the output qubit of $V$,
which without loss of generality is assumed to be the first qubit of $\regV$.
$V$ performs a unitary quantum computation over ${(\regV, \regM)}$
and then measures the output qubit in the computational basis,
where the measurement outcome~$1$ corresponds to acceptance.
On an input~$x$ in $\Sigma^\ast$,
the number of private qubits in $\regV$
and the length of a quantum witness in $\regM$
are restricted to~${l_\regV(\abs{x})}$ and ${l_\regM(\abs{x})}$
according to some predetermined functions~$l_\regV$~and~$l_\regM$
that depend only on the input length~$\abs{x}$.
Unless explicitly mentioned,
no restriction is put on the time complexity of the unitary quantum computation of $V$.

Formally,
for functions~$\function{l_\regV, l_\regM}{\Nonnegative}{\Natural}$,
an
\emph{%
  ${(l_\regV, l_\regM)}$-space-bounded quantum verifier~$V$
  for a space-bounded unitary quantum Merlin-Arthur proof system%
}
is a machine that on an input~$x$ in $\Sigma^\ast$ performs a unitary transformation~$V_x$,
where each $V_x$ acts over~${\op{l_\regV}(\abs{x}) + \op{l_\regM}(\abs{x})}$~qubits,
the first ${\op{l_\regV}(\abs{x})}$~qubits of which correspond to the register~$\regV$
and the rest ${\op{l_\regM}(\abs{x})}$~qubits of which correspond to the register~$\regM$.
It is assumed that such a machine~$V$ corresponds to
a certain reasonable $l$-space-bounded unitary quantum computation model
for some function~$\function{l}{\Nonnegative}{\Natural}$
such that ${\op{l}(n)}$ is in ${\op{O}(\op{l_\regV}(n) + \op{l_\regM}(n))}$.
For instance, $V$ may be an $l$-space classical-quantum hybrid Turing machine~\cite{Wat03CC, Wat09ECSS}
for unitary quantum computations,
or may be a machine
that first runs a classical Turing machine of
an $l$-space uniformly generated family of unitary quantum circuits
and then performs the generated circuit.
It is stressed that all the results in this paper hold
regardless of the models of space-bounded quantum computations
as long as the computations performed are unitary.

Fix an input~$x$ in $\Sigma^\ast$,
and suppose that $V$ receives a quantum witness~$\rho$ of ${\op{l_\regM}(\abs{x})}$~qubits in $\regM$.
The probability~${\op{p_\acc}(V_x, \rho)}$ that $V$ accepts $x$ with a quantum witness~$\rho$
is given by
\[
  \op{p_\acc}(V_x, \rho)
  =
  \tr \Pi_\acc
      \conjugate{V_x} \bigl[ (\ketbra{0})^{\tensor \op{l_\regV}(\abs{x})} \tensor \rho \bigr] V_x,
\]
where ${\Pi_\acc = \ketbra{1} \tensor I^{\tensor (\op{l_\regV}(\abs{x}) + \op{l_\regM}(\abs{x}) - 1)}}$
is the projection onto the subspace spanned by the states
in which the designated output qubit is in state~$\ket{1}$.

The class~${\unitaryQMASPACE[l_\regV, l_\regM](c, s)}$ of problems
having ${(l_\regV, l_\regM)}$-space-bounded unitary QMA systems
is defined as follows.

\begin{definition}
  Given functions~$\function{l_\regV, l_\regM}{\Nonnegative}{\Natural}$
  and $\function{c,s}{\Nonnegative}{[0,1]}$ satisfying ${c > s}$,
  a promise problem~${A = (A_\yes, A_\no)}$ is in ${\unitaryQMASPACE[l_\regV, l_\regM](c, s)}$
  if there exists an ${(l_\regV, l_\regM)}$-space-bounded quantum verifier~$V$
  for a space-bounded unitary quantum Merlin-Arthur proof system
  such that, for every $x$ in $\Sigma^\ast$,
  \begin{description}
  \item[\textnormal{(Completeness)}]
    if $x$ is in $A_\yes$,
    there exists a quantum witness~$\rho$ of ${\op{l_\regM}(\abs{x})}$~qubits
    that makes $V$ accept $x$ with probability at least ${\op{c}(\abs{x})}$,
    and
  \item[\textnormal{(Soundness)}]
    if $x$ is in $A_\no$,
    for any quantum witness~$\rho$ of ${\op{l_\regM}(\abs{x})}$~qubits,
    $V$ accepts $x$ with probability at most ${\op{s}(\abs{x})}$.
  \end{description}
  \label{Definition: QMASPACE(l_V, l_M, c, s)}
\end{definition}

Note that quantum witnesses may be restricted to pure states,
as allowing quantum witnesses of mixed states
does not increase the maximal accepting probability of proof systems.

The classes~${\unitaryQMAL(c, s)}$~and~${\unitaryQMAPSPACE(c, s)}$
corresponding to the logarithmic-space and polynomial-space QMA-type computations, respectively,
with completeness~$c$ and soundness~$s$
are then obtained by restricting both of the functions~$l_\regV$~and~$l_\regM$
in Definition~\ref{Definition: QMASPACE(l_V, l_M, c, s)}
to be logarithmically bounded and polynomially bounded.
\begin{sloppy}
\begin{definition}
  Given functions~$\function{c,s}{\Nonnegative}{[0,1]}$ satisfying ${c > s}$,
  a promise problem~${A = (A_\yes, A_\no)}$ is in ${\unitaryQMAL(c, s)}$
  iff $A$ is in ${\unitaryQMASPACE[l_\regV, l_\regM](c, s)}$
  for some logarithmically bounded functions~$\function{l_\regV, l_\regM}{\Nonnegative}{\Natural}$.
  \label{Definition: QMAL(c, s)}
\end{definition}

\begin{definition}
  Given functions~$\function{c,s}{\Nonnegative}{[0,1]}$ satisfying ${c > s}$,
  a promise problem~${A = (A_\yes, A_\no)}$ is in ${\unitaryQMAPSPACE(c, s)}$
  iff $A$ is in ${\unitaryQMASPACE[l_\regV, l_\regM](c, s)}$
  for some polynomially bounded functions~$\function{l_\regV, l_\regM}{\Nonnegative}{\Natural}$.
  \label{Definition: QMAPSPACE(c, s)}
\end{definition}

When ${l_\regM = 0}$ in Definitions~\ref{Definition: QMAL(c, s)}~and~\ref{Definition: QMAPSPACE(c, s)},
respectively,
the resulting classes~${\unitaryQL(c, s)}$~and~${\unitaryQPSPACE(c, s)}$
correspond to the standard logarithmic-space and polynomial-space unitary quantum computations
with completeness~$c$ and soundness~$s$.

\begin{definition}
  Given functions~$\function{c,s}{\Nonnegative}{[0,1]}$ satisfying ${c > s}$,
  a promise problem~${A = (A_\yes, A_\no)}$ is in ${\unitaryQL(c, s)}$
  iff $A$ is in ${\unitaryQMASPACE[l_\regV, 0](c, s)}$
  for some logarithmically bounded function~$\function{l_\regV}{\Nonnegative}{\Natural}$.
  \label{Definition: QL(c, s)}
\end{definition}

\begin{definition}
  Given functions~$\function{c,s}{\Nonnegative}{[0,1]}$ satisfying ${c > s}$,
  a promise problem~${A = (A_\yes, A_\no)}$ is in ${\unitaryQPSPACE(c, s)}$
  iff $A$ is in ${\unitaryQMASPACE[l_\regV, 0](c, s)}$
  for some polynomially bounded function~$\function{l_\regV}{\Nonnegative}{\Natural}$.
  \label{Definition: QPSPACE(c, s)}
\end{definition}

Finally, the bounded-error classes~$\unitaryQMAL$~and~$\unitaryBQL$
may be defined as follows.

\begin{definition}
  A promise problem~${A = (A_\yes, A_\no)}$ is in $\unitaryQMAL$
  iff $A$ is in ${\unitaryQMAL(2/3, 1/3)}$.
  \label{Definition: QMAL}
\end{definition}

\begin{definition}
  A promise problem~${A = (A_\yes, A_\no)}$ is in $\unitaryBQL$
  iff $A$ is in ${\unitaryQL(2/3, 1/3)}$.
  \label{Definition: BQL}
\end{definition}
\end{sloppy}


\section{Basic Procedures}
\label{Section: basic procedures}

Let $\spaceH$ be any Hilbert space of dimension a power of two.
Given a unitary transformation~$U$ and two projections~$\Delta$~and~$\Pi$,
all acting over $\spaceH$,
define the Hermitian operator~$M$ over $\spaceH$ by
\[
  M = \Delta \conjugate{U} \Pi U \Delta,
\]
which plays crucial roles
in many well-known amplification methods in quantum computation,
including
the Grover search~\cite{Gro96STOC},
the Marriott-Watrous amplification for $\QMA$~\cite{MarWat05CC},
the Nagaj-Wocjan-Zhang amplification for $\QMA$ based on phase estimation~\cite{NagWocZha09QIC},
and quantum rewinding for zero-knowledge proofs against quantum attacks~\cite{Wat09SIComp}.


\paragraph{\textsc{One-Shot Phase-Estimation Procedure}}

Consider the procedure described in Figure~\ref{Figure: One-Shot Phase-Estimation Procedure},
which is at the core of the amplification method
based on phase estimation proposed by Nagaj,~Wocjan,~and~Zhang~\cite{NagWocZha09QIC}.
\begin{figure}[t!]
  \begin{algorithm*}
        {
          \textsc{One-Shot Phase-Estimation Procedure}
          associated with $\boldsymbol{(U, \Delta, \Pi, t, l, \varepsilon)}$
        }
    \begin{step}
    \item
      Receive a quantum register~$\regQ$
      that contains a state in the subspace corresponding to the projection~$\Delta$.
    \item
      Let $Q$ be the unitary transformation defined by
      ${Q = (2 \conjugate{U} \Pi U - I_\regQ)(2 \Delta - I_\regQ)}$.
      Perform the phase estimation associated with $Q$ acting over the state in $\regQ$
      with precision of $l$~bits and failure probability~$\varepsilon$,
      using ${l + \bigceil{\log \bigl( 2 + \frac{1}{2 \varepsilon} \bigr)}}$~ancilla qubits.
      Accept if the estimated phase is in the interval~${(-t, t)}$
      and reject otherwise.
    \end{step}
  \end{algorithm*}
  \caption{
    The \textsc{One-Shot Phase-Estimation Procedure}.
  }
  \label{Figure: One-Shot Phase-Estimation Procedure}
\end{figure}
The following proposition holds with the \textsc{One-Shot Phase-Estimation Procedure}.

\begin{proposition}[\cite{NagWocZha09QIC}]
  Let $U$ be a unitary transformation and $\Delta$~and~$\Pi$ be projections,
  all acting over the same Hilbert space.
  Let $\varepsilon$ be a real number in ${(0,1)}$,
  let $l$ be a positive integer,
  and let $t$ be a real number in ${\bigl[ 0, \frac{1}{2} \bigr]}$ represented by $l$~bits.
  Consider the Hermitian operator~${M = \Delta \conjugate{U} \Pi U \Delta}$.
  The following two properties hold:
  \begin{description}
  \item[\textnormal{(Completeness)}]
    Suppose that $M$ has an eigenstate~$\ket{\phi_\lambda}$
    with its associated eigenvalue~$\lambda$ satisfying that
    ${\frac{1}{\pi} \arccos \sqrt{\lambda} \leq t-2^{-l}}$.
    Then,
    the \textsc{One-Shot Phase-Estimation Procedure}
    associated with ${(U, \Delta, \Pi, t, l, \varepsilon)}$
    results in acceptance with probability~${1 - \varepsilon}$
    when the state~$\ket{\phi_\lambda}$ is received in register~$\regQ$ in Step~1.
  \item[\textnormal{(Soundness)}]
    Suppose that all the eigenvalues~$\lambda$ of $M$ are such that
    ${\frac{1}{\pi} \arccos \sqrt{\lambda} \geq t + 2^{-l}}$.
    Then,
    the \textsc{One-Shot Phase-Estimation Procedure}
    associated with ${(U, \Delta, \Pi, t, l, \varepsilon)}$
    results in acceptance with probability at most $\varepsilon$
    regardless of the quantum state received in register~$\regQ$ in Step~1.
  \end{description}
  \label{Proposition: properties of One-Shot Phase-Estimation Procedure}
\end{proposition}

\begin{remark}
  One thing to be mentioned is
  that the standard phase-estimation algorithm involves inverting quantum Fourier transformation,
  which cannot be implemented exactly
  when implemented by quantum circuits with a gate set of finite size.
  Thus, one needs to approximately implement some phase-rotation gates.
  The number of phase-rotation gates necessary to approximate is proportional to $l^2$
  to achieve precision of $l$~bits
  in the standard implementation of a phase-estimation algorithm.
  This means that each phase-rotation gate must be approximated
  within ${\op{O} \bigl( \frac{\varepsilon}{l^2} \bigr)}$
  so that approximate implementation does not significantly affect
  the failure probability~$\varepsilon$ of the phase-estimation algorithm.
  To prove Theorem~\ref{Theorem: main theorem} via the simple construction based on phase estimation,
  one needs to perform a phase-estimation algorithm
  with precision~$l$ at least logarithmic with respect to $\frac{p}{c - s}$
  and with failure probability~$\varepsilon$ at most polynomially small with respect to $p$.
  The standard (constructive) proofs of the Solovay-Kitaev theorem~\cite{Kit97RMS}
  (such as those found in Refs.~\cite{NieChu00Book, KitSheVya02Book, DawNie06QIC})
  require space polylogarithmic with respect to $\frac{1}{\delta}$
  when approximating within $\delta$,
  which is insufficient for the purpose of proving Theorem~\ref{Theorem: main theorem}
  via the simple construction based on phase estimation.
  Fortunately, van~Melkebeek~and~Watson~\cite{MelWat12ToC}
  showed a more space-efficient construction of the Solovay-Kitaev approximation,
  which uses space only logarithmic with respect to $\frac{1}{\delta}$
  and can be used for the simple construction based on phase estimation
  to prove Theorem~\ref{Theorem: main theorem}.
\end{remark}


\paragraph{\textsc{AND-Type Repetition Procedure}}

Given a unitary transformation~$U$ and two projections~$\Delta$~and~$\Pi$
all acting over a Hilbert space,
consider the process of applying $U$
to a fixed initial state~$\ket{\phi}$ in a quantum register~$\regQ$
that is in the subspace corresponding to $\Delta$
and then accepting if and only if
the resulting state is projected onto the subspace corresponding to $\Pi$
by the projective measurement~${\{\Pi, I_\regQ - \Pi\}}$.
Let $p$ denote the accepting probability of this process.
By running $N$~independent attempts of such a process,
the probability clearly becomes $p^N$
for the event that all the attempts result in acceptance,
but which requires $N$~copies of the initial state~$\ket{\phi}$.
When $\ket{\phi}$ is an eigenstate of the Hermitian operator~${M = \Delta \conjugate{U} \Pi U \Delta}$,
the following \textsc{AND-Type Repetition Procedure}
essentially simulates such independent attempts with just one copy of $\ket{\phi}$.

Prepare an $l$-qubit register~$\regC$ that serves as a counter modulo $2^l$,
where ${l = \ceil{\log (2N + 1)}}$.
All the qubits in $\regC$ are initialized to state~$\ket{0}$.
The procedure receives a quantum register~$\regQ$
that contains a state in the subspace corresponding to $\Delta$,
and then repeats $N$~times a pair of
a simulation attempt by $U$ and an initialization attempt by $\conjugate{U}$.
After each attempt of applying $U$ to $\regQ$,
the procedure checks
if the state in $\regQ$ belongs to the subspace corresponding to $\Pi$,
and increments the counter in $\regC$ if this check fails.
Similarly, after each attempt of applying $\conjugate{U}$ to $\regQ$,
it checks if the state in $\regQ$ is back to a legal initial state
belonging to the subspace corresponding to $\Delta$,
and increments the counter in $\regC$ if this check fails.
After the repetition, the procedure accepts if and only if the counter in $\regC$ is still $0$.
Figure~\ref{Figure: AND-Type Repetition Procedure}
presents the precise description of the \textsc{AND-Type Repetition Procedure}.
\begin{figure}[t!]
  \begin{algorithm*}
        {\textsc{AND-Type Repetition Procedure} associated with $\boldsymbol{(U, \Delta, \Pi, N)}$}
    \begin{step}
    \item
      Let ${l = \ceil{\log (2N + 1)}}$, and prepare an $l$-qubit register~$\regC$,
      where all the qubits in $\regC$ are initialized to state~$\ket{0}$.
      Receive a quantum register~$\regQ$
      that contains a state in the subspace corresponding to the projection~$\Delta$.
    \item
      For ${j = 1}$ to $N$, perform the following:
      \begin{step}
      \item
        Apply $U$ to $\regQ$.
      \item
        If the state in $\regQ$ belongs to the subspace
        corresponding to the projection~${I_\regQ - \Pi}$,
        apply $\INCR{2^l}$ to $\regC$,
        where $\INCR{2^l}$ is the unitary transformation defined by
        \[
          \INCR{2^l} \colon \ket{j} \mapsto \bigket{(j+1) \bmod 2^l},
          \quad
          \forall j \in \Integers_{2^l}.
        \]
      \item
        Apply $\conjugate{U}$ to $\regQ$.
      \item
        If the state in $\regQ$ belongs to the subspace
        corresponding to the projection~${I_\regQ - \Delta}$,
        apply $\INCR{2^l}$ to $\regC$.
      \end{step}
    \item
      Accept if the content of $\regC$ is $0$
      (i.e., all the qubits in $\regC$ are in state~$\ket{0}$),
      and reject otherwise.
    \end{step}
  \end{algorithm*}
  \caption{
    The \textsc{AND-Type Repetition Procedure}.
  }
  \label{Figure: AND-Type Repetition Procedure}
\end{figure}

The following proposition holds with the \textsc{AND-Type Repetition Procedure}.

\begin{proposition}
  Let $U$ be a unitary transformation and $\Delta$~and~$\Pi$ be projections,
  all acting over the same Hilbert space,
  and let $N$ be a positive integer.
  For the \textsc{AND-Type Repetition Procedure} associated with ${(U, \Delta, \Pi, N)}$,
  let $U'$ be the unitary transformation induced by it,
  let $\Delta'$ be the projection onto the subspace spanned by the legal initial states of it,
  and let $\Pi'$ be the projection onto the subspace spanned by the accepting states of it.
  Suppose that
  the Hermitian operator~${M = \Delta \conjugate{U} \Pi U \Delta}$
  has an eigenstate~$\ket{\phi_\lambda}$
  with its associated eigenvalue~$\lambda$.
  Then the state~${\ket{\phi_\lambda} \tensor \ket{0}^{\tensor l}}$
  is an eigenstate of
  the Hermitian operator~${M' = \Delta' \conjugate{(U')} \Pi' U' \Delta'}$
  with eigenvalue~$\lambda^{2N}$.
  \label{Proposition: eigenstate of AND-Type Repetition Procedure}
\end{proposition}

\begin{proof}
  The unitary transformation~$U'$ can be written as
  \[
    U'
    =
    \bigl\{
      \bigl[ \Delta \tensor I_\regC + (I_\regQ - \Delta) \tensor \INCR{2^l} \bigr]
      (\conjugate{U} \otimes I_\regC)
      \bigl[ \Pi \tensor I_\regC + (I_\regQ - \Pi) \tensor \INCR{2^l} \bigr]
      (U \otimes I_\regC)
    \bigr\}^N,
  \]
  whereas
  the projections~$\Delta'$~and~$\Pi'$ can be written as
  \[
    \Delta' = \Delta \tensor (\ketbra{0})^{\tensor l},
    \quad
    \Pi' = I_\regQ \tensor (\ketbra{0})^{\tensor l}.
  \]
  Notice that, for any $k$ in ${\{1, \dotsc, 2N\}}$, it holds that
  \[
    (\ketbra{0})^{\tensor l} \bigl( \INCR{2^l} \bigr)^k (\ketbra{0})^{\tensor l} = 0,
  \]
  since the content of $\regC$, which starts at $0$,
  cannot return to $0$ for $k$~applications of the increment transformation~$\INCR{2^l}$,
  for ${k \leq 2N < 2^l}$.
  This implies that $M'$ can be simply written as
  \[
    M'
    =
    \Delta' \conjugate{(U')} \Pi' U' \Delta'
    =
    \bigl[
      \Delta
      \bigl[ \conjugate{(\Delta \conjugate{U} \Pi U)} \bigr]^N
      (\Delta \conjugate{U} \Pi U)^N
      \Delta
    \bigr]
    \tensor
    (\ketbra{0})^{\tensor l}
    =
    M^{2N} \tensor (\ketbra{0})^{\tensor l}.
  \]
  Hence,
  if $\ket{\phi_\lambda}$ is an eigenstate of $M$ with eigenvalue~$\lambda$,
  then ${\ket{\phi_\lambda} \tensor \ket{0}^{\tensor l}}$ is an eigenstate of $M'$ with eigenvalue~$\lambda^{2N}$.
\end{proof}

Now the following property of the \textsc{AND-Type Repetition Procedure}
is immediate from Proposition~\ref{Proposition: eigenstate of AND-Type Repetition Procedure}.

\begin{proposition}
  Let $U$ be a unitary transformation and $\Delta$~and~$\Pi$ be projections,
  all acting over the same Hilbert space,
  and let $N$ be a positive integer.
  Consider the Hermitian operator~${M = \Delta \conjugate{U} \Pi U \Delta}$.
  The following two properties hold:
  \begin{description}
  \item[\textnormal{(Completeness)}]
    Suppose that $M$ has an eigenstate~$\ket{\phi_\lambda}$
    with its associated eigenvalue~$\lambda$.
    Then,
    the \textsc{AND-Type Repetition Procedure} associated with ${(U, \Delta, \Pi, N)}$
    results in acceptance with probability~$\lambda^{2N}$
    when the state~$\ket{\phi_\lambda}$ is received in register~$\regQ$ in Step~1.
  \item[\textnormal{(Soundness)}]
    Suppose that all the eigenvalues of $M$ are at most $\varepsilon$
    for some $\varepsilon$ in ${[0, 1)}$.
    Then,
    the \textsc{AND-Type Repetition Procedure} associated with ${(U, \Delta, \Pi, N)}$
    results in acceptance with probability at most $\varepsilon^{2N}$
    regardless of the quantum state received in register~$\regQ$ in Step~1.
  \end{description}
  \label{Proposition: properties of AND-Type Repetition Procedure}
\end{proposition}


\paragraph{\textsc{OR-Type Repetition Procedure}}

One can also construct a procedure
that essentially simulates the process of taking OR of the $N$~independent attempts mentioned before
with just one copy of $\ket{\phi}$.
One now applies $\INCR{2^l}$ to $\regC$
when the state in $\regQ$ belongs to the subspace
corresponding to the projection~$\Pi$
at Step~2.2 of the \textsc{AND-Type Repetition Procedure},
and \emph{rejects}
if and only if the content of $\regC$ is $0$
at Step~3 of the \textsc{AND-Type Repetition Procedure}.
The resulting procedure is called the \textsc{OR-Type Repetition Procedure},
whose precise description is presented in Figure~\ref{Figure: OR-Type Repetition Procedure}.
\begin{figure}[t!]
  \begin{algorithm*}
        {\textsc{OR-Type Repetition Procedure} associated with $\boldsymbol{(U, \Delta, \Pi, N)}$}
    \begin{step}
    \item
      Let ${l = \ceil{\log (2N + 1)}}$, and prepare an $l$-qubit register~$\regC$,
      where all the qubits in $\regC$ are initialized to state~$\ket{0}$.
      Receive a quantum register~$\regQ$
      that contains a state in the subspace corresponding to the projection~$\Delta$.
    \item
      For ${j = 1}$ to $N$, perform the following:
      \begin{step}
      \item
        Apply $U$ to $\regQ$.
      \item
        If the state in $\regQ$ belongs to the subspace
        corresponding to the projection~$\Pi$,
        apply $\INCR{2^l}$ to $\regC$,
        where $\INCR{2^l}$ is the unitary transformation defined by
        \[
          \INCR{2^l} \colon \ket{j} \mapsto \bigket{(j+1) \bmod 2^l},
          \quad
          \forall j \in \Integers_{2^l}.
        \]
      \item
        Apply $\conjugate{U}$ to $\regQ$.
      \item
        If the state in $\regQ$ belongs to the subspace
        corresponding to the projection~${I_\regQ - \Delta}$,
        apply $\INCR{2^l}$ to $\regC$.
      \end{step}
    \item
      Reject if the content of $\regC$ is $0$
      (i.e., all the qubits in $\regC$ are in state~$\ket{0}$),
      and accept otherwise.
    \end{step}
  \end{algorithm*}
  \caption{
    The \textsc{OR-Type Repetition Procedure}.
  }
  \label{Figure: OR-Type Repetition Procedure}
\end{figure}

Similarly to the \textsc{AND-Type Repetition Procedure},
the following proposition holds with the \textsc{OR-Type Repetition Procedure}.

\begin{proposition}
  Let $U$ be a unitary transformation and $\Delta$~and~$\Pi$ be projections,
  all acting over the same Hilbert space,
  and let $N$ be a positive integer.
  For the \textsc{OR-Type Repetition Procedure} associated with ${(U, \Delta, \Pi, N)}$,
  let $U'$ be the unitary transformation induced by it,
  let $\Delta'$ be the projection onto the subspace spanned by the legal initial states of it,
  and let $\Pi'$ be the projection onto the subspace spanned by the accepting states of it.
  Suppose that
  the Hermitian operator~${M = \Delta \conjugate{U} \Pi U \Delta}$
  has an eigenstate~$\ket{\phi_\lambda}$
  with its associated eigenvalue~$\lambda$.
  Then the state~${\ket{\phi_\lambda} \tensor \ket{0}^{\tensor l}}$
  is an eigenstate of
  the Hermitian operator~${M' = \Delta' \conjugate{(U')} \Pi' U' \Delta'}$
  with eigenvalue~${1 - (1 - \lambda)^{2N}}$.
  \label{Proposition: eigenstate of OR-Type Repetition Procedure}
\end{proposition}

\begin{proof}
  The proof is very similar to the proof of Proposition~\ref{Proposition: eigenstate of AND-Type Repetition Procedure}.
  This time, the unitary transformation~$U'$ can be written as
  \[
    U'
    =
    \bigl\{
      \bigl[ \Delta \tensor I_\regC + (I_\regQ - \Delta) \tensor \INCR{2^l} \bigr]
      (\conjugate{U} \otimes I_\regC)
      \bigl[ \Pi \tensor \INCR{2^l} + (I_\regQ - \Pi) \tensor I_\regC \bigr]
      (U \otimes I_\regC)
    \bigr\}^N,
  \]
  whereas
  the projections~$\Delta'$~and~$\Pi'$ can be written as
  \[
    \Delta' = \Delta \tensor (\ketbra{0})^{\tensor l},
    \quad
    \Pi' = I_\regQ \tensor \bigl[ I_\regC - (\ketbra{0})^{\tensor l} \bigr].
  \]
  Again notice that, for any $k$ in ${\{1, \dotsc, 2N\}}$, it holds that
  \[
    (\ketbra{0})^{\tensor l} \bigl( \INCR{2^l} \bigr)^k (\ketbra{0})^{\tensor l} = 0,
  \]
  and thus, $M'$ can be simply written as
  \[
    \begin{split}
      \hspace{5mm}
      &
      \hspace{-5mm}
      M'
      =
      \Delta'
      -
      \Delta' \conjugate{(U')} (I_{(\regQ, \regC)} - \Pi') U' \Delta'
      \\
      &
      =
      \bigl\{
        \Delta
        -
        \Delta
        \bigl[ \conjugate{(\Delta \conjugate{U} (I_\regQ - \Pi) U)} \bigr]^N
        [\Delta \conjugate{U} (I_\regQ - \Pi) U]^N
        \Delta
      \bigr\}
      \tensor
      (\ketbra{0})^{\tensor l}
      \\
      &
      =
      \bigl[ \Delta - (\Delta - M)^{2N} \bigr] \tensor (\ketbra{0})^{\tensor l}.
    \end{split}
  \]
  Now notice that
  ${
    \lambda \ket{\phi_\lambda}
    =
    M \ket{\phi_\lambda}
    =
    \Delta M \ket{\phi_\lambda}
    =
    \lambda \Delta \ket{\phi_\lambda}
  }$,
  and therefore
  at least one of ${\Delta \ket{\phi_\lambda} = \ket{\phi_\lambda}}$ or ${\lambda = 0}$ holds.
  If ${\Delta \ket{\phi_\lambda} = \ket{\phi_\lambda}}$, it obviously holds that 
  \[
    M' \bigl( \ket{\phi_\lambda} \tensor \ket{0}^{\tensor l} \bigr)
    =
    \bigl[ 1 - (1 - \lambda)^{2N} \bigr] \bigl( \ket{\phi_\lambda} \tensor \ket{0}^{\tensor l} \bigr).
  \]
  On the other hand, when ${\lambda = 0}$,
  by using that ${M \Delta = \Delta M = M}$ and ${M \ket{\phi_\lambda} = 0}$,
  it follows that
  \[
    M' \bigl( \ket{\phi_\lambda} \tensor \ket{0}^{\tensor l} \bigr)
    =
    \bigl( \Delta - \Delta^{2N} \bigr) \bigl( \ket{\phi_\lambda} \tensor \ket{0}^{\tensor l} \bigr)
    =
    0,
  \]
  which is sufficient for the claim, because ${1 - (1 - \lambda)^{2N} = 0}$ in this case.
\end{proof}

Now the following property of the \textsc{OR-Type Repetition Procedure}
is immediate from Proposition~\ref{Proposition: eigenstate of OR-Type Repetition Procedure}.

\begin{proposition}
  Let $U$ be a unitary transformation and $\Delta$~and~$\Pi$ be projections,
  all acting over the same Hilbert space,
  and let $N$ be a positive integer.
  Consider the Hermitian operator~${M = \Delta \conjugate{U} \Pi U \Delta}$.
  The following two properties hold:
  \begin{description}
  \item[\textnormal{(Completeness)}]
    Suppose that $M$ has an eigenstate~$\ket{\phi_\lambda}$
    with its associated eigenvalue~$\lambda$.
    Then,
    the \textsc{OR-Type Repetition Procedure} associated with ${(U, \Delta, \Pi, N)}$
    results in acceptance with probability~${1 - (1 - \lambda)^{2N}}$
    when the state~$\ket{\phi_\lambda}$ is received in register~$\regQ$ in Step~1.
  \item[\textnormal{(Soundness)}]
    Suppose that all the eigenvalues of $M$ are at most $\varepsilon$
    for some $\varepsilon$ in ${[0, 1)}$.
    Then,
    the \textsc{OR-Type Repetition Procedure} associated with ${(U, \Delta, \Pi, N)}$
    results in acceptance with probability at most ${1 - (1 - \varepsilon)^{2N}}$
    regardless of the quantum state received in register~$\regQ$ in Step~1.
  \end{description}
  \label{Proposition: properties of OR-Type Repetition Procedure}
\end{proposition}


\paragraph{\textsc{Marriott-Watrous Amplification Procedure}}

Consider the procedure described in Figure~\ref{Figure: Marriott-Watrous Amplification Procedure},
which is exactly the amplification method (described in a general form)
proposed by Marriott~and~Watrous~\cite{MarWat05CC}.
\begin{figure}[t!]
  \begin{algorithm*}{\textsc{Marriott-Watrous Amplification Procedure} associated with $\boldsymbol{(U, \Delta, \Pi, N, t)}$}
    \begin{step}
    \item
      Let ${l = \ceil{\log (2N + 1)}}$.
      Prepare a single-qubit register~$\regB_j$ for each $j$ in ${\{0, \dotsc, 2N\}}$,
      and an $l$-qubit register~$\regC$,
      where all the qubits in $\regB_j$ and $\regC$ are initialized to state~$\ket{0}$.
      Receive a quantum register~$\regQ$
      that contains a state in the subspace corresponding to the projection~$\Delta$.
    \item
      For ${j = 1}$ to $N$, perform the following:
      \begin{step}
      \item
        Apply $U$ to $\regQ$.
      \item
        If the state in $\regQ$ belongs to the subspace
        corresponding to the projection~${I_\regQ - \Pi}$,
        apply the Pauli transformation~$X$ (i.e., the $\NOT$ transformation) to $\regB_j$.
      \item
        Apply $\conjugate{U}$ to $\regQ$.
      \item
        If the state in $\regQ$ belongs to the subspace
        corresponding to the projection~${I_\regQ - \Delta}$,
        apply $X$ to $\regB_{j+1}$.
      \end{step}
    \item
      For ${j = 1}$ to ${2N}$, perform the following:\\
      If the content of $\regB_j$ is the same as that of $\regB_{j-1}$,
      apply $\INCR{2^l}$ to $\regC$,
      where $\INCR{2^l}$ is the unitary transformation defined by
      \[
        \INCR{2^l} \colon \ket{j} \mapsto \bigket{(j+1) \bmod 2^l},
        \quad
        \forall j \in \Integers_{2^l}.
      \]
    \item
      Accept if the content of $\regC$ is at least~$t$
      (when viewed as an integer in $\Integers_{2^l}$),
      and reject otherwise.
    \end{step}
  \end{algorithm*}
  \caption{
    The \textsc{Marriott-Watrous Amplification Procedure}.
  }
  \label{Figure: Marriott-Watrous Amplification Procedure}
\end{figure}

The following proposition holds with the \textsc{Marriott-Watrous Amplification Procedure}.

\begin{proposition}[\cite{MarWat05CC}]
  Let $U$ be a unitary transformation and $\Delta$~and~$\Pi$ be projections,
  all acting over the same Hilbert space.
  Let $N$~and~$t$ be positive integers satisfying ${t \leq 2N}$.
  Consider the Hermitian operator~${M = \Delta \conjugate{U} \Pi U \Delta}$.
  The following two properties hold:
  \begin{description}
  \item[\textnormal{(Completeness)}]
    Suppose that $M$ has an eigenstate~$\ket{\phi_\lambda}$
    with its associated eigenvalue~${\lambda \geq \frac{t}{2N} + \varepsilon}$
    for some $\varepsilon$ in ${\bigl( 0, 1 - \frac{t}{2N} \bigr]}$.
    Then,
    the \textsc{Marriott-Watrous Amplification Procedure} associated with ${(U, \Delta, \Pi, N, t)}$
    results in acceptance with probability greater than~${1 - e^{- 4 \varepsilon^2 N}}$
    when the state~$\ket{\phi_\lambda}$ is received in register~$\regQ$ in Step~1.
  \item[\textnormal{(Soundness)}]
    Suppose that all the eigenvalues of $M$ are at most~${\frac{t}{2N} - \varepsilon}$
    for some $\varepsilon$ in ${\bigl( 0, \frac{t}{2N} \bigr]}$.
    Then,
    the \textsc{Marriott-Watrous Amplification Procedure} associated with ${(U, \Delta, \Pi, N, t)}$
    results in acceptance with probability less than~${e^{- 4 \varepsilon^2 N}}$
    regardless of the quantum state received in register~$\regQ$ in Step~1.
  \end{description}
  \label{Proposition: properties of Marriott-Watrous Amplification Procedure}
\end{proposition}


\paragraph{\textsc{Additive Adjustment Procedure}}

For a Hilbert space~$\spaceH_j$ for each $j$ in ${\{1, 2\}}$,
consider a unitary transformation~$U_j$ and two projections~$\Delta_j$~and~$\Pi_j$,
all acting over $\spaceH_j$.
Define the Hermitian operator~$M_j$ over $\spaceH_j$ for each $j$ in ${\{1, 2\}}$ by
${
  M_j = \Delta_j \conjugate{U_j} \Pi_j U_j \Delta_j
}$.

Now define a Hilbert space~$\spaceH'$ defined by
${
  \spaceH' = \spaceB \tensor \spaceH_1 \tensor \spaceH_2
}$,
where ${\spaceB = \op{\Complex}(\Sigma)}$ is a Hilbert space corresponding to a single qubit.
Let
\[
  \Delta'
  =
  \ketbra{0} \tensor \Delta_1 \tensor \Delta_2,
  \quad
  \Pi'
  =
  \ketbra{0} \tensor \Pi_1 \tensor I_{\spaceH_2}
  +
  \ketbra{1} \tensor I_{\spaceH_1} \tensor \Pi_2,
  \quad
  U'
  =
  H \tensor U_1 \tensor U_2,
\]
where $H$ denotes the Hadamard transformation,
and further let
${
  M' = \Delta' \conjugate{(U')} \Pi' U' \Delta'
}$.
A straightforward calculation shows that
\[
  M'
  =
  \frac{1}{2} (\ketbra{0} \tensor M_1 \tensor \Delta_1 + \ketbra{0} \tensor \Delta_2 \tensor M_2).
\]
Suppose that,
for each $j$ in ${\{1, 2\}}$,
the Hermitian operator~$M_j$ has an eigenstate (i.e., the normalized eigenvector)~$\ket{\phi_{j, \lambda_j}}$
with its associated eigenvalue~$\lambda_j$.
It is easy to see that
\[
  M' (\ket{0} \tensor \ket{\phi_{1, \lambda_1}} \tensor \ket{\phi_{2, \lambda_2}})
  =
  \frac{\lambda_1 + \lambda_2}{2}
  (\ket{0} \tensor \ket{\phi_{1, \lambda_1}} \tensor \ket{\phi_{2, \lambda_2}}).
\]
This implies that $M'$ has an eigenstate~%
${\ket{0} \tensor \ket{\phi_{1,\lambda_1}} \tensor \ket{\phi_{2,\lambda_2}}}$
with eigenvalue~$\frac{\lambda_1 + \lambda_2}{2}$,
which is implicit in the additive adjustment technique of acceptance probability proposed in Ref.~\cite{JorKobNagNis12QIC}.
This leads to the following \textsc{Additive Adjustment Procedure}
presented in Figure~\ref{Figure: Additive Adjustment Procedure}.
%
\begin{figure}[t!]
  \begin{algorithm*}
        {\textsc{Additive Adjustment Procedure} associated with $\boldsymbol{(U, \Delta, \Pi, l, k)}$}
    \begin{step}
    \item
      Prepare a single-qubit register~$\regB$
      and an $l$-qubit register~$\regR$,
      where all the qubits in $\regB$ and $\regR$ are initialized to state~$\ket{0}$.
      Receive a quantum register~$\regQ$
      that contains a state in the subspace corresponding to the projection~$\Delta$.
    \item
      Apply the Hadamard transformation~$H$ to each qubit in ${(\regB, \regR)}$,
      and apply $U$ to $\regQ$.
    \item
      Accept either if $\regB$ contains $0$
      \emph{and} the state in $\regQ$ belongs to the subspace corresponding to $\Pi$
      or if $\regB$ contains $1$ \emph{and} the content of $\regR$ is greater than $k$
      (when viewed as an integer in ${\{1, \dotsc, 2^l\}}$),
      and reject otherwise.
    \end{step}
  \end{algorithm*}
  \caption{
    The \textsc{Additive Adjustment Procedure}.
  }
  \label{Figure: Additive Adjustment Procedure}
\end{figure}

The following proposition is immediate from the argument above.

\begin{proposition}
  Let $U$ be a unitary transformation and $\Delta$~and~$\Pi$ be projections,
  all acting over the same Hilbert space,
  and let $l$ be a positive integer and $k$ be an integer in ${\{1, \dotsc, 2^l\}}$.
  For the \textsc{Additive Adjustment Procedure} associated with ${(U, \Delta, \Pi, l, k)}$,
  let $U'$ be the unitary transformation induced by it,
  let $\Delta'$ be the projection onto the subspace spanned by the legal initial states of it,
  and let $\Pi'$ be the projection onto the subspace spanned by the accepting states of it.
  Suppose that
  the Hermitian operator~${M = \Delta \conjugate{U} \Pi U \Delta}$
  has an eigenstate~$\ket{\phi_\lambda}$
  with its associated eigenvalue~$\lambda$.
  Then the state~${\ket{0} \tensor \ket{\phi_\lambda} \tensor \ket{0}^{\tensor l}}$
  is an eigenstate of
  the Hermitian operator~${M' = \Delta' \conjugate{(U')} \Pi' U' \Delta'}$
  with eigenvalue~${\frac{1}{2} + \frac{1}{2} \bigl( \lambda - \frac{k}{2^l} \bigr)}$.
  \label{Proposition: eigenstate of Additive Adjustment Procedure}
\end{proposition}

Now the following property of the \textsc{Additive Adjustment Procedure}
is immediate from Proposition~\ref{Proposition: eigenstate of Additive Adjustment Procedure}.

\begin{proposition}
  Let $U$ be a unitary transformation and $\Delta$~and~$\Pi$ be projections,
  all acting over the same Hilbert space.
  Consider the Hermitian operator~${M = \Delta \conjugate{U} \Pi U \Delta}$.
  For any positive integer~$l$ and any integer~$k$ in ${\{1, \dotsc, 2^l\}}$,
  the following two properties hold:
  \begin{description}
  \item[\textnormal{(Completeness)}]
    Suppose that $M$ has an eigenstate~$\ket{\phi_\lambda}$
    with its associated eigenvalue~$\lambda$.
    Then,
    the \textsc{Additive Adjustment Procedure} associated with ${(U, \Delta, \Pi, l, k)}$
    results in acceptance with probability~%
    ${\frac{1}{2} + \frac{1}{2} \bigl( \lambda - \frac{k}{2^l} \bigr)}$
    when the state~$\ket{\phi_\lambda}$ is received in register~$\regQ$ in Step~1.
  \item[\textnormal{(Soundness)}]
    Suppose that all the eigenvalues of $M$ are at most $\varepsilon$
    for some $\varepsilon$ in ${[0, 1)}$.
    Then,
    the \textsc{Additive Adjustment Procedure} associated with ${(U, \Delta, \Pi, l, k)}$
    results in acceptance with probability at most
    ${
      \frac{1}{2} + \frac{1}{2} \bigl( \varepsilon - \frac{k}{2^l} \bigr)
    }$
    regardless of the quantum state received in register~$\regQ$ in Step~1.
  \end{description}
  \label{Proposition: properties of Additive Adjustment Procedure}
\end{proposition}


\paragraph{\textsc{Reflection Procedure}}

Finally, consider the procedure described in Figure~\ref{Figure: Reflection Procedure in main},
which is exactly the \textsc{Reflection Procedure} in a general form
originally developed in Ref.~\cite{KobLeGNis15SIComp}.
\begin{figure}[t!]
  \begin{algorithm*}{\textsc{Reflection Procedure} associated with $\boldsymbol{(U, \Delta, \Pi)}$}
    \begin{step}
    \item
      Receive a quantum register~$\regQ$
      that contains a state in the subspace corresponding to the projection~$\Delta$.
    \item
      Apply $U$ to $\regQ$.
    \item
      Perform a phase-flip
      (i.e., multiply the phase by $-1$)
      if the state in $\regQ$ belongs to the subspace corresponding to the projection~$\Pi$.
    \item
      Apply $\conjugate{U}$ to $\regQ$.
    \item
      Reject if the state in $\regQ$ belongs to the subspace corresponding to $\Delta$,
      and accept otherwise.
    \end{step}
  \end{algorithm*}
  \caption{
    The \textsc{Reflection Procedure}.
  }
  \label{Figure: Reflection Procedure in main}
\end{figure}

The following proposition holds with the \textsc{Reflection Procedure}.

\begin{proposition}[\cite{KobLeGNis15SIComp}]
  Let $U$ be a unitary transformation and $\Delta$~and~$\Pi$ be projections,
  all acting over the same Hilbert space.
  Consider the Hermitian operator~${M = \Delta \conjugate{U} \Pi U \Delta}$.
  The following two properties hold:
  \begin{description}
  \item[\textnormal{(Completeness)}]
    Suppose that $M$ has an eigenstate~$\ket{\phi_\lambda}$
    with its associated eigenvalue~$\lambda$.
    Then,
    the \textsc{Reflection Procedure} associated with ${(U, \Delta, \Pi)}$
    results in acceptance with probability~${4 \lambda (1 - \lambda)}$
    when the state~$\ket{\phi_\lambda}$ is received in register~$\regQ$ in Step~1.
  \item[\textnormal{(Soundness)}]
    Suppose that none of the eigenvalues of $M$
    is in the interval~${\bigl( \frac{1}{2} - \varepsilon, \frac{1}{2} + \varepsilon \bigr)}$
    for some $\varepsilon$ in ${\bigl( 0, \frac{1}{2} \bigr]}$.
    Then,
    the \textsc{Reflection Procedure} associated with ${(U, \Delta, \Pi)}$
    results in acceptance with probability at most ${1 - 4 \varepsilon^2}$
    regardless of the quantum state received in register~$\regQ$ in Step~1.
  \end{description}
  \label{Proposition: properties of Reflection Procedure}
\end{proposition}


\section{Space-Efficient Amplification Methods}
\label{Section: space-efficient amplification methods}

This section rigorously proves Theorem~\ref{Theorem: main theorem} in the three different ways.

Throughout this section,
consider any QMA-type computation for a problem~${A = (A_\yes, A_\no)}$
induced by a family~${\{V_x\}_{x \in \Sigma^\ast}}$
of a unitary transformation~$V_x$ of the verifier on input~$x$ in $\Sigma^\ast$
that acts over a quantum register~${\regQ = (\regV, \regM)}$,
where $\regV$ is the quantum register consisting of all the private qubits of the verifier,
and $\regM$ is the one for storing a received quantum witness.
Let $\Pi_\init$ be the projection onto the subspace spanned by
the legal initial states of the QMA-type computation induced by $V_x$
(i.e., the subspace spanned by those in which all the qubits in $\regV$ is in state~$\ket{0}$)
and let $\Pi_\acc$ be the projection onto the subspace spanned by
the accepting states of the QMA-type computation associated with $V_x$
(i.e., the subspace spanned by states for which the designated output qubit of $V_x$ is in state~$\ket{0}$).
The maximum eigenvalue of the Hermitian operator~%
${M_x = \Pi_\init \conjugate{V_x} \Pi_\acc V_x \Pi_\init}$
exactly corresponds to
the maximum acceptance probability of the verifier on input~$x$
over all possible quantum witnesses received in $\regM$.
Hence, $M_x$ has an eigenvalue at least ${\op{c}(\abs{x})}$ if $x$ is in $A_\yes$,
while all eigenvalues of $M_x$ are at most ${\op{s}(\abs{x})}$ if $x$ is in $A_\no$,
where $\function{c,s}{\Nonnegative}{[0,1]}$ are functions
that provide completeness and soundness conditions of
the QMA-type computation induced by ${\{V_x\}_{x \in \Sigma^\ast}}$, respectively.


\subsection{Simple Construction Based on Phase Estimation}
\label{Subsection: simple construction based on phase estimation}

The first proof is via the simple construction based on phase estimation.


\paragraph{Mild amplification with a phase estimation}

Fix a function~$\function{p}{\Nonnegative}{\Natural}$
and functions~$\function{c,s}{\Nonnegative}{[0,1]}$ satisfying ${c > s}$,
arbitrarily.
Let $\function{l}{\Nonnegative}{\Natural}$ be a function defined by
\[
  l = \biggceil{\log \frac{2 \pi}{\arccos \sqrt{s} - \arccos \sqrt{c}}},
\]
and let $\function{t}{\Nonnegative}{\bigl[ 0, \frac{1}{2} \bigr]}$ be a function
such that, for every nonnegative integer~$n$, ${\op{t}(n)}$ is an approximation of
${\frac{1}{2 \pi} \bigl( \arccos \sqrt{\op{c}(n)} + \arccos \sqrt{\op{s}(n)} \bigr)}$
with ${\op{l}(n)}$-bit precision.

Fix an input~$x$.
Given the triplet~${(V_x, \Pi_\init, \Pi_\acc)}$,
one constructs the \textsc{One-Shot Phase-Estimation Procedure}
associated with
${\bigl( V_x, \Pi_\init, \Pi_\acc, \op{t}(\abs{x}), \op{l}(\abs{x}), \frac{1}{\op{p}(\abs{x})} \bigr)}$.
The resulting procedure is called the \textsc{Mild Amplification with Phase Estimation},
and is summarized in Figure~\ref{Figure: Mild Amplification with Phase Estimation}.
\begin{figure}[t!]
  \begin{algorithm*}
        {
          \textsc{Mild Amplification with Phase Estimation}
          associated with $\boldsymbol{(V_x, p)}$
        }
    \begin{flushright}
      \begin{minipage}{0.9727\textwidth}
        Define a function~$\function{l}{\Nonnegative}{\Natural}$ by
        ${
          l = \bigceil{\log \frac{2 \pi}{\arccos \sqrt{s} - \arccos \sqrt{c}}}
        }$
        and let $\function{t}{\Nonnegative}{\bigl[ 0, \frac{1}{2} \bigr]}$ be a function
        such that, for every nonnegative integer~$n$, ${\op{t}(n)}$ is an approximation of
        ${\frac{1}{2 \pi} \bigl( \arccos \sqrt{\op{c}(n)} + \arccos \sqrt{\op{s}(n)} \bigr)}$
        with ${\op{l}(n)}$-bit precision.
        Let $\Pi_\init$ and $\Pi_\acc$ be the projections onto the subspaces
        spanned by the legal initial states and the accepting states, respectively,
        in the verification with $V_x$.\\
        Perform\hspace{-0.05mm} the\hspace{-0.05mm} \textsc{One-Shot\hspace{-0.05mm} Phase-Estimation\hspace{-0.05mm} Procedure}\hspace{-0.05mm}
        associated\hspace{-0.05mm} with\hspace{-0.46mm}
        ${\bigl( V_x,\hspace{-0.46mm} \Pi_\init,\hspace{-0.46mm} \Pi_\acc,\hspace{-0.46mm} \op{t}(\abs{x}),\hspace{-0.39mm} \op{l}(\abs{x}),\hspace{-0.46mm} \frac{1}{\op{p}(\abs{x})} \bigr)}$.
      \end{minipage}
    \end{flushright}
    \vspace{0.75\baselineskip}
  \end{algorithm*}
  \caption{
    The \textsc{Mild Amplification with Phase Estimation}.
  }
  \label{Figure: Mild Amplification with Phase Estimation}
\end{figure}

The following lemma is proved by using the \textsc{Mild Amplification with Phase Estimation}
combined with the properties of the \textsc{One-Shot Phase-Estimation Procedure}
stated in Proposition~\ref{Proposition: properties of One-Shot Phase-Estimation Procedure}.


\begin{lemma}
  For any functions~$\function{p, l_\regV, l_\regM}{\Nonnegative}{\Natural}$
  and any functions~$\function{c,s}{\Nonnegative}{[0,1]}$
  satisfying
  ${c > s}$,
  there exists a function~$\function{\delta}{\Nonnegative}{\Natural}$
  that is logarithmic with respect to ${\frac{p}{c - s}}$
  such that
  \[
    \unitaryQMASPACE[l_\regV, l_\regM](c, s)
    \subseteq
    \unitaryQMASPACE[l_\regV + \delta, l_\regM] \biggl( 1 - \frac{1}{p}, \frac{1}{p} \biggr).
  \]
  \label{Lemma: mild amplification}
\end{lemma}


\begin{proof}
  Let ${A = (A_\yes, A_\no)}$ be a problem in ${\unitaryQMASPACE[l_\regV, l_\regM](c, s)}$,
  and let ${V = \{ V_x \}_{x \in \Sigma^\ast}}$
  be the ${(l_\regV, l_\regM)}$-space-bounded quantum verifier witnessing this membership.
  Fix a function~$\function{p}{\Nonnegative}{\Natural}$ and an input~$x$ in $\Sigma^\ast$.
  Consider the \textsc{One-Shot Phase-Estimation Procedure}
  associated with
  ${\bigl( V_x, \Pi_\init, \Pi_\acc, \op{t}(\abs{x}), \op{l}(\abs{x}), \frac{1}{\op{p}(\abs{x})} \bigr)}$,
  which is exactly
  what the \textsc{Mild Amplification with Phase Estimation} associated with ${(V_x, p)}$ performs.

  From Proposition~\ref{Proposition: properties of One-Shot Phase-Estimation Procedure},
  it holds that,
  if $x$ is in $A_\yes$,
  the \textsc{One-Shot Phase-Estimation Procedure}
  associated with
  ${\bigl( V_x, \Pi_\init, \Pi_\acc, \op{t}(\abs{x}), \op{l}(\abs{x}), \frac{1}{\op{p}(\abs{x})} \bigr)}$
  results in acceptance with probability at least~${1 - \frac{1}{\op{p}(\abs{x})}}$,
  while if $x$ is in $A_\no$,
  it results in acceptance with probability at most~$\frac{1}{\op{p}(\abs{x})}$,
  which shows the completeness and soundness.

  The \textsc{One-Shot Phase-Estimation Procedure}
  associated with
  ${\bigl( V_x, \Pi_\init, \Pi_\acc, \op{t}(\abs{x}), \op{l}(\abs{x}), \frac{1}{\op{p}(\abs{x})} \bigr)}$
  uses extra workspace of
  ${\op{\delta}(\abs{x}) = \op{l}(\abs{x}) + \bigceil{\log \bigl( \frac{\op{p}(\abs{x})}{2} + 2 \bigr)}}$~qubits.
  As is proved in Ref~\cite{NagWocZha09QIC},
  the function~${l = \bigceil{\log \frac{2 \pi}{\arccos \sqrt{s} - \arccos \sqrt{c}}}}$
  is logarithmic with respect to $\frac{1}{c-s}$,
  and thus,
  the used extra workspace is logarithmic with respect to $\frac{p}{c-s}$, as claimed.
\end{proof}


\paragraph{Soundness error-reduction}

Again fix arbitrarily a function~$\function{p}{\Nonnegative}{\Natural}$
and functions~$\function{c,s}{\Nonnegative}{[0,1]}$ satisfying ${c > s}$,
and let $\function{N}{\Nonnegative}{\Natural}$ be a function defined by
\[
  N = \Bigceil{\frac{p}{2 \log (2p + 4)}}.
\]

Fix an input~$x$.
Given the pair~${(V_x, p)}$,
consider the \textsc{Mild Amplification with Phase Estimation}
associated with ${(V_x, 2p + 4)}$.
Let $V'_x$ be the unitary transformation induced by it,
let $\Pi'_\init$ be the projection onto the subspace spanned by the legal initial states of it,
and let $\Pi'_\acc$ be the projection onto the subspace spanned by the accepting states of it.
From the triplet~${\bigl( V'_x, \Pi'_\init, \Pi'_\acc \bigr)}$
and a positive integer~${\op{N}(\abs{x})}$,
one constructs the \textsc{AND-Type Repetition Procedure}
associated with ${\bigl( V'_x, \Pi'_\init, \Pi'_\acc, \op{N}(\abs{x}) \bigr)}$,
and performs it.
The resulting procedure is called the \textsc{Soundness Error Reduction},
and is summarized in Figure~\ref{Figure: Soundness Error Reduction}.
\begin{figure}[t!]
  \begin{algorithm*}
        {
          \textsc{Soundness Error Reduction}
          associated with $\boldsymbol{(V_x, p)}$
        }
    \begin{flushright}
      \begin{minipage}{0.9727\textwidth}
        Define a function~$\function{N}{\Nonnegative}{\Natural}$ by
        ${N = \bigceil{\frac{p}{2 \log (2p + 4)}}}$.
        Consider the \textsc{Mild Amplification with Phase Estimation}
        associated with ${(V_x, 2p + 4)}$.
        Let $V'_x$ be the unitary transformation induced by it,
        let $\Pi'_\init$ be the projection onto the subspace spanned by the legal initial states of it,
        and let $\Pi'_\acc$ be the projection onto the subspace spanned by the accepting states of it.\\
        Perform the \textsc{AND-Type Repetition Procedure}
        associated with ${\bigl( V'_x, \Pi'_\init, \Pi'_\acc, \op{N}(\abs{x}) \bigr)}$.
      \end{minipage}
    \end{flushright}
    \vspace{0.75\baselineskip}
  \end{algorithm*}
  \caption{
    The \textsc{Soundness Error Reduction}.
  }
  \label{Figure: Soundness Error Reduction}
\end{figure}

The following lemma is proved by using the \textsc{Soundness Error Reduction}
combined with the properties of the \textsc{Mild Amplification with Phase Estimation}
used for proving Lemma~\ref{Lemma: mild amplification}.

\begin{lemma}
  For any functions~$\function{p, l_\regV, l_\regM}{\Nonnegative}{\Natural}$
  and any functions~$\function{c,s}{\Nonnegative}{[0,1]}$
  satisfying
  ${c > s}$,
  there exists a function~$\function{\delta}{\Nonnegative}{\Natural}$
  that is logarithmic with respect to ${\frac{p}{c - s}}$
  such that
  \[
    \unitaryQMASPACE[l_\regV, l_\regM](c, s)
    \subseteq
    \unitaryQMASPACE[l_\regV + \delta, l_\regM] \biggl( \frac{1}{2}, 2^{-p} \biggr).
  \]
  \label{Lemma: soundness error reduction}
\end{lemma}

\begin{proof}
  Let ${A = (A_\yes, A_\no)}$ be a problem in ${\unitaryQMASPACE[l_\regV, l_\regM](c, s)}$,
  and let ${V = \{ V_x \}_{x \in \Sigma^\ast}}$
  be the ${(l_\regV, l_\regM)}$-space-bounded quantum verifier witnessing this membership.
  Fix a function~$\function{p}{\Nonnegative}{\Natural}$ and an input~$x$ in $\Sigma^\ast$.
  The lemma is proved by considering the \textsc{Soundness Error Reduction} associated with ${(V_x, p)}$.

  First consider the \textsc{Mild Amplification with Phase Estimation}
  associated with ${(V_x, 2p + 4)}$.
  Let $V'_x$ be the unitary transformation induced by it,
  let $\Pi'_\init$ be the projection onto the subspace spanned by the legal initial states of it,
  and let $\Pi'_\acc$ be the projection onto the subspace spanned by the accepting states of it.

  Lemma~\ref{Lemma: mild amplification} and its proof
  ensure that
  $A$ is in
  ${\unitaryQMASPACE[l_\regV + \delta_1, l_\regM] \bigl( 1 - \frac{1}{2p + 4}, \frac{1}{2p + 4} \bigr)}$
  for some function~$\function{\delta_1}{\Nonnegative}{\Natural}$
  that is logarithmic with respect to ${\frac{p}{c - s}}$,
  and this inclusion is certified by
  the \textsc{Mild Amplification with Phase Estimation}
  associated with ${(V_x, 2p + 4)}$.
  This in particular implies that
  the Hermitian operator~%
  ${M'_x = \Pi'_\init \conjugate{(V'_x)} \Pi'_\acc V'_x \Pi'_\init}$
  has an eigenvalue at least~%
  ${1 - \frac{1}{2 \op{p}(\abs{x}) + 4}}$
  if $x$ is in $A_\yes$,
  while all the eigenvalues of $M'_x$ are at most~%
  $\frac{1}{2 \op{p}(\abs{x}) + 4}$
  if $x$ is in $A_\no$.

  Now consider the \textsc{AND-Type Repetition Procedure}
  associated with ${\bigl( V'_x, \Pi'_\init, \Pi'_\acc, \op{N}(\abs{x}) \bigr)}$,
  which is exactly what the \textsc{Soundness Error Reduction} associated with ${(V_x, p)}$ performs.
  By Proposition~\ref{Proposition: properties of AND-Type Repetition Procedure},
  the \textsc{AND-Type Repetition Procedure}
  associated with ${\bigl( V'_x, \Pi'_\init, \Pi'_\acc, \op{N}(\abs{x}) \bigr)}$
  results in acceptance with probability at least
  \[
    \biggl( 1 - \frac{1}{2 \op{p}(\abs{x}) + 4} \biggr)^{2 \op{N}(\abs{x})}
    >
    \biggl( 1 - \frac{1}{2 \op{p}(\abs{x}) + 4} \biggr)^{\op{p}(\abs{x}) + 2}
    >
    \frac{1}{2}
  \]
  if $x$ is in $A_\yes$,
  and at most
  \[
    \biggl( \frac{1}{2 \op{p}(\abs{x}) + 4} \biggr)^{2 \op{N}(\abs{x})}
    \leq
    \Bigl( 2^{- \log (2 \op{p}(\abs{x}) + 4)} \Bigr)^{\frac{\op{p}(\abs{x})}{\log (2 \op{p}(\abs{x}) + 4)}}
    =
    2^{- \op{p}(\abs{x})}
  \]
  if $x$ is in $A_\no$,
  and the completeness and soundness follows.

  The \textsc{AND-Type Repetition Procedure}
  associated with
  ${\bigl( V'_x, \Pi'_\init, \Pi'_\acc, \op{N}(\abs{x}) \bigr)}$
  uses extra workspace
  (relative to $V'_x$)
  of ${\op{\delta}_2 (\abs{x})}$~qubits
  for the function~$\function{\delta_2}{\Nonnegative}{\Natural}$ defined by
  ${\delta_2 = \ceil{\log (2N + 1)}}$.
  As ${N = \bigceil{\frac{p}{2 \log (2p + 4)}}}$,
  $\delta_2$ is clearly logarithmic with respect to $p$, and thus, with respect to $\frac{p}{c-s}$ also.
  Hence, the \textsc{Soundness Error Reduction} associated with ${(V_x, p)}$
  uses extra workspace
  (relative to $V_x$)
  of logarithmically many qubits with respect to $\frac{p}{c-s}$ also
  (which is determined by a function~${\delta = \delta_1 + \delta_2}$),
  as desired.
\end{proof}


\paragraph{Space-efficient error reduction based on phase estimation}

Again fix arbitrarily a function~$\function{p}{\Nonnegative}{\Natural}$
and functions~$\function{c,s}{\Nonnegative}{[0,1]}$ satisfying ${c > s}$,
and let $\function{N}{\Nonnegative}{\Natural}$ be a function defined by
\[
  N = \Bigceil{\frac{p}{2}}.
\]

Fix an input~$x$.
Given the pair~${(V_x, p)}$,
consider the \textsc{Soundness Error Reduction}
associated with ${\bigl( V_x, p + \ceil{\log (p + 2)} \bigr)}$.
Let $V'_x$ be the unitary transformation induced by it,
let $\Pi'_\init$ be the projection onto the subspace spanned by the legal initial states of it,
and let $\Pi'_\acc$ be the projection onto the subspace spanned by the accepting states of it.
From the triplet~${\bigl( V'_x, \Pi'_\init, \Pi'_\acc \bigr)}$
and a positive integer~${\op{N}(\abs{x})}$,
one constructs the \textsc{OR-Type Repetition Procedure}
associated with ${\bigl( V'_x, \Pi'_\init, \Pi'_\acc, \op{N}(\abs{x}) \bigr)}$,
and performs it.
The resulting procedure is called
the \textsc{Space-Efficient Error Reduction Based on Phase Estimation},
and is summarized in Figure~\ref{Figure: Space-Efficient Error Reduction Based on Phase Estimation}.
\begin{figure}[t!]
  \begin{algorithm*}
        {
          \textsc{Space-Efficient Error Reduction Based on Phase Estimation}
          associated with $\boldsymbol{(V_x, p)}$
        }
    \begin{flushright}
      \begin{minipage}{0.9727\textwidth}
        Define a function~$\function{N}{\Nonnegative}{\Natural}$ by
        ${N = \bigceil{\frac{p}{2}}}$.
        Consider the \textsc{Soundness Error Reduction}
        associated with ${\bigl( V_x, p + \ceil{\log (p + 2)} \bigr)}$.
        Let $V'_x$ be the unitary transformation induced by it,
        let $\Pi'_\init$ be the projection onto the subspace spanned by the legal initial states of it,
        and let $\Pi'_\acc$ be the projection onto the subspace spanned by the accepting states of it.\\
        Perform the \textsc{OR-Type Repetition Procedure}
        associated with ${\bigl( V'_x, \Pi'_\init, \Pi'_\acc, \op{N}(\abs{x}) \bigr)}$.
      \end{minipage}
    \end{flushright}
    \vspace{0.75\baselineskip}
  \end{algorithm*}
  \caption{
    The \textsc{Space-Efficient Error Reduction Based on Phase Estimation}.
  }
  \label{Figure: Space-Efficient Error Reduction Based on Phase Estimation}
\end{figure}

Now Theorem~\ref{Theorem: main theorem},
the main theorem of this paper,
is ready to be proved
by using the \textsc{Space-Efficient Error Reduction Based on Phase Estimation}
combined with the properties of the \textsc{Soundness Error Reduction}
used for proving Lemma~\ref{Lemma: soundness error reduction}.

\begin{proof}
      [%
        Proof of Theorem~\ref{Theorem: main theorem}
        (via the simple construction based on phase estimation)
      ]
  Let ${A = (A_\yes, A_\no)}$ be a problem in ${\unitaryQMASPACE[l_\regV, l_\regM](c, s)}$,
  and let ${V = \{ V_x \}_{x \in \Sigma^\ast}}$
  be the ${(l_\regV, l_\regM)}$-space-bounded quantum verifier witnessing this membership.
  Fix a function~$\function{p}{\Nonnegative}{\Natural}$ and an input~$x$ in $\Sigma^\ast$.
  The theorem is proved by considering
  the \textsc{Space-Efficient Error Reduction Based on Phase Estimation} associated with ${(V_x, p)}$.

  First consider the \textsc{Soundness Error Reduction}
  associated with ${\bigl( V_x, p + \ceil{\log (p + 2)} \bigr)}$.
  Let $V'_x$ be the unitary transformation induced by it,
  let $\Pi'_\init$ be the projection onto the subspace spanned by the legal initial states of it,
  and let $\Pi'_\acc$ be the projection onto the subspace spanned by the accepting states of it.

  Lemma~\ref{Lemma: soundness error reduction} and its proof
  ensure that
  $A$ is in
  ${\unitaryQMASPACE[l_\regV + \delta_1, l_\regM] \bigl( \frac{1}{2}, \frac{1}{p + 2} \cdot 2^{-p} \bigr)}$
  for some function~$\function{\delta_1}{\Nonnegative}{\Natural}$
  that is logarithmic with respect to ${\frac{p}{c - s}}$,
  and this inclusion is certified by
  the \textsc{Soundness Error Reduction}
  associated with ${\bigl( V_x, p + \ceil{\log (p + 2)} \bigr)}$.
  This in particular implies that
  the Hermitian operator~%
  ${M'_x = \Pi'_\init \conjugate{(V'_x)} \Pi'_\acc V'_x \Pi'_\init}$
  has an eigenvalue at least~%
  $\frac{1}{2}$
  if $x$ is in $A_\yes$,
  while all the eigenvalues of $M'_x$ are at most~%
  ${\frac{1}{\op{p}(\abs{x}) + 2} \cdot 2^{- \op{p}(\abs{x})}}$
  if $x$ is in $A_\no$.

  Now consider the \textsc{OR-Type Repetition Procedure}
  associated with ${\bigl( V'_x, \Pi'_\init, \Pi'_\acc, \op{N}(\abs{x}) \bigr)}$,
  which is exactly what
  the \textsc{Space-Efficient Error Reduction Based on Phase Estimation} associated with ${(V_x, p)}$
  performs.
  By Proposition~\ref{Proposition: properties of OR-Type Repetition Procedure},
  the \textsc{OR-Type Repetition Procedure}
  associated with ${\bigl( V'_x, \Pi'_\init, \Pi'_\acc, \op{N}(\abs{x}) \bigr)}$
  results in acceptance with probability at least
  \[
    1 - \Bigl( 1 - \frac{1}{2} \Bigr)^{2 \op{N}(\abs{x})}
    \geq
    1 - 2^{- \op{p}(\abs{x})}
  \]
  if $x$ is in $A_\yes$,
  and at most
  \[
    1 - \biggl( 1 - \frac{1}{\op{p}(\abs{x}) + 2} \cdot 2^{- \op{p}(\abs{x})} \biggr)^{2 \op{N}(\abs{x})}
    <
    1 - \biggl( 1 - \frac{1}{\op{p}(\abs{x}) + 2} \cdot 2^{- \op{p}(\abs{x})} \biggr)^{\op{p}(\abs{x}) + 2}
    <
    2^{- \op{p}(\abs{x})}
  \]
  if $x$ is in $A_\no$,
  and the completeness and soundness follows.

  The \textsc{OR-Type Repetition Procedure}
  associated with
  ${\bigl( V'_x, \Pi'_\init, \Pi'_\acc, \op{N}(\abs{x}) \bigr)}$
  uses extra workspace
  (relative to $V'_x$)
  of ${\op{\delta}_2 (\abs{x})}$~qubits
  for the function~$\function{\delta_2}{\Nonnegative}{\Natural}$ defined by
  ${\delta_2 = \ceil{\log (2N + 1)}}$.
  As ${N = \bigceil{\frac{p}{2}}}$,
  $\delta_2$ is clearly logarithmic with respect to $p$, and thus, with respect to $\frac{p}{c-s}$ also.
  Hence,
  the \textsc{Space-Efficient Error Reduction Based on Phase Estimation} associated with ${(V_x, p)}$
  uses extra workspace
  (relative to $V_x$)
  of logarithmically many qubits with respect to $\frac{p}{c-s}$ also
  (which is determined by a function~${\delta = \delta_1 + \delta_2}$),
  as desired.
\end{proof}

Recall that
the necessary number of calls of the (controlled) unitary transformation~$U$ is
${2^l \cdot \bigceil{\frac{1}{2 \varepsilon} + 2} - 1}$
for a phase estimation associated with $U$
precise to $l$~bits with failure probability~$\varepsilon$~\cite{NieChu00Book}.
Hence, a straightforward calculation shows that
the \textsc{Space-Efficient Error Reduction Based on Phase Estimation} associated with ${(V_x, p)}$ uses
${\op{O} \bigl( \frac{1}{c - s} \cdot \frac{p^3}{\log p} \bigr)}$~calls
of $V_x$ and its inverse.


\subsection{Hybrid construction of phase estimation and Marriott-Watrous}
\label{Subsection: hybrid construction of phase estimation and Marriott-Watrous}

The second proof is based on the hybrid construction of phase estimation and Marriott-Watrous.


\paragraph{Very mild amplification with a phase estimation}

Fix functions~$\function{c,s}{\Nonnegative}{[0,1]}$ satisfying ${c > s}$,
arbitrarily.
Again let $\function{l}{\Nonnegative}{\Natural}$ be a function defined by
\[
  l = \biggceil{\log \frac{2 \pi}{\arccos \sqrt{s} - \arccos \sqrt{c}}},
\]
and let $\function{t}{\Nonnegative}{\bigl[ 0, \frac{1}{2} \bigr]}$ be a function
such that, for every nonnegative integer~$n$, ${\op{t}(n)}$ is an approximation of
${\frac{1}{2 \pi} \bigl( \arccos \sqrt{\op{c}(n)} + \arccos \sqrt{\op{s}(n)} \bigr)}$
with ${\op{l}(n)}$-bit precision.

Fix an input~$x$.
Given the triplet~${(V_x, \Pi_\init, \Pi_\acc)}$,
one constructs the \textsc{One-Shot Phase-Estimation Procedure}
associated with
${\bigl( V_x, \Pi_\init, \Pi_\acc, \op{t}(\abs{x}), \op{l}(\abs{x}), \frac{1}{4} \bigr)}$.
The resulting procedure is called the \textsc{Very Mild Amplification with Phase Estimation},
and is summarized in Figure~\ref{Figure: Very Mild Amplification with Phase Estimation}.
\begin{figure}[t!]
  \begin{algorithm*}
        {
          \textsc{Very Mild Amplification with Phase Estimation}
          associated with $\boldsymbol{V_x}$
        }
    \begin{flushright}
      \begin{minipage}{0.9727\textwidth}
        Define a function~$\function{l}{\Nonnegative}{\Natural}$ by
        ${
          l = \bigceil{\log \frac{2 \pi}{\arccos \sqrt{s} - \arccos \sqrt{c}}}
        }$
        and let $\function{t}{\Nonnegative}{\bigl[ 0, \frac{1}{2} \bigr]}$ be a function
        such that, for every nonnegative integer~$n$, ${\op{t}(n)}$ is an approximation of
        ${\frac{1}{2 \pi} \bigl( \arccos \sqrt{\op{c}(n)} + \arccos \sqrt{\op{s}(n)} \bigr)}$
        with ${\op{l}(n)}$-bit precision.
        Let $\Pi_\init$ and $\Pi_\acc$ be the projections onto the subspaces
        spanned by the legal initial states and the accepting states, respectively,
        in the verification with $V_x$.\\
        Perform the \textsc{One-Shot Phase-Estimation Procedure}
        associated with
        ${\bigl( V_x, \Pi_\init, \Pi_\acc, \op{t}(\abs{x}), \op{l}(\abs{x}), \frac{1}{4} \bigr)}$.
      \end{minipage}
    \end{flushright}
    \vspace{0.75\baselineskip}
  \end{algorithm*}
  \caption{
    The \textsc{Very Mild Amplification with Phase Estimation}.
  }
  \label{Figure: Very Mild Amplification with Phase Estimation}
\end{figure}

In fact, the \textsc{Very Mild Amplification with Phase Estimation} associated with $V_x$
is nothing but the \textsc{Mild Amplification with Phase Estimation} associated with ${(V_x, 4)}$.
Hence,
the following lemma is immediate by using the \textsc{Very Mild Amplification with Phase Estimation}
combined with Lemma~\ref{Lemma: mild amplification} and its proof.

\begin{lemma}
  For any functions~$\function{l_\regV, l_\regM}{\Nonnegative}{\Natural}$
  and any functions~$\function{c,s}{\Nonnegative}{[0,1]}$
  satisfying
  ${c > s}$,
  there exists a function~$\function{\delta}{\Nonnegative}{\Natural}$
  that is logarithmic with respect to ${\frac{1}{c - s}}$
  such that
  \[
    \unitaryQMASPACE[l_\regV, l_\regM](c, s)
    \subseteq
    \unitaryQMASPACE[l_\regV + \delta, l_\regM] \biggl( \frac{3}{4}, \frac{1}{4} \biggr).
  \]
  \label{Lemma: very mild amplification with a phase estimation}
\end{lemma}

\paragraph{Mild amplification with Marriott-Watrous}

Fix a function~$\function{p}{\Nonnegative}{\Natural}$
and functions~$\function{c,s}{\Nonnegative}{[0,1]}$ satisfying ${c > s}$,
arbitrarily.
Let $\function{N}{\Nonnegative}{\Natural}$ be a function defined by
\[
  N = \biggceil{\frac{4 \log p}{\log e}}.
\]

Fix an input~$x$.
Given the pair~${(V_x, p)}$,
consider the \textsc{Very Mild Amplification with Phase Estimation} associated with $V_x$.
Let $V'_x$ be the unitary transformation induced by it,
let $\Pi'_\init$ be the projection onto the subspace spanned by the legal initial states of it,
and let $\Pi'_\acc$ be the projection onto the subspace spanned by the accepting states of it.
From the triplet~${\bigl( V'_x, \Pi'_\init, \Pi'_\acc \bigr)}$
and a positive integer~${\op{N}(\abs{x})}$,
one constructs the \textsc{Marriott-Watrous Amplification Procedure}
associated with ${\bigl( V'_x, \Pi'_\init, \Pi'_\acc, \op{N}(\abs{x}) \bigr)}$,
and performs it.
The resulting procedure is called the \textsc{Mild Amplification with Marriott-Watrous},
and is summarized in Figure~\ref{Figure: Mild Amplification with Marriott-Watrous}.
\begin{figure}[t!]
  \begin{algorithm*}
        {
          \textsc{Mild Amplification with Marriott-Watrous}
          associated with $\boldsymbol{(V_x, p)}$
        }
    \begin{flushright}
      \begin{minipage}{0.9727\textwidth}
        Define a function~$\function{N}{\Nonnegative}{\Natural}$ by
        ${N = \bigceil{\frac{4 \log p}{\log e}}}$.
        Consider the \textsc{Very Mild Amplification with Phase Estimation} associated with $V_x$.
        Let $V'_x$ be the unitary transformation induced by it,
        let $\Pi'_\init$ be the projection onto the subspace spanned by the legal initial states of it,
        and let $\Pi'_\acc$ be the projection onto the subspace spanned by the accepting states of it.\\
        Perform the \textsc{Marriott-Watrous Amplification Procedure}
        associated with ${\bigl( V'_x, \Pi'_\init, \Pi'_\acc, \op{N}(\abs{x}), \op{N}(\abs{x}) \bigr)}$.
      \end{minipage}
    \end{flushright}
    \vspace{0.75\baselineskip}
  \end{algorithm*}
  \caption{
    The \textsc{Mild Amplification with Marriott-Watrous}.
  }
  \label{Figure: Mild Amplification with Marriott-Watrous}
\end{figure}

Now Lemma~\ref{Lemma: mild amplification}
is alternatively proved by using the \textsc{Mild Amplification with Marriott-Watrous}
combined with the properties of the \textsc{Marriott-Watrous Amplification Procedure}
stated in Proposition~\ref{Proposition: properties of Marriott-Watrous Amplification Procedure}.


\begin{proof}
      [%
        Proof of Lemma~\ref{Lemma: mild amplification}
        (via the hybrid construction of phase estimation and Marriott-Watrous)
      ]
  Let ${A\hspace{-0.072mm} =\hspace{-0.072mm} (A_\yes,\hspace{-0.072mm} A_\no)}$
  be a problem in ${\unitaryQMASPACE[l_\regV, l_\regM](c, s)}$,
  and let ${V = \{ V_x \}_{x \in \Sigma^\ast}}$
  be the ${(l_\regV, l_\regM)}$-space-bounded quantum verifier witnessing this membership.
  Fix a function~$\function{p}{\Nonnegative}{\Natural}$ and an input~$x$ in $\Sigma^\ast$.
  The lemma is proved by considering
  the \textsc{Mild Amplification with Marriott-Watrous} associated with ${(V_x, p)}$.

  First consider the \textsc{Very Mild Amplification with Phase Estimation} associated with $V_x$.
  Let $V'_x$ be the unitary transformation induced by it,
  let $\Pi'_\init$ be the projection onto the subspace spanned by the legal initial states of it,
  and let $\Pi'_\acc$ be the projection onto the subspace spanned by the accepting states of it.

  Lemma~\ref{Lemma: very mild amplification with a phase estimation} and its proof
  ensure that
  $A$ is in
  ${\unitaryQMASPACE[l_\regV + \delta_1, l_\regM] \bigl( \frac{3}{4}, \frac{1}{4} \bigr)}$
  for some function~$\function{\delta_1}{\Nonnegative}{\Natural}$
  that is logarithmic with respect to ${\frac{1}{c - s}}$,
  and this inclusion is certified by
  the \textsc{Very Mild Amplification with Phase Estimation} associated with $V_x$.
  This in particular implies that
  the Hermitian operator~%
  ${M'_x = \Pi'_\init \conjugate{(V'_x)} \Pi'_\acc V'_x \Pi'_\init}$
  has an eigenvalue at least~$\frac{3}{4}$
  if $x$ is in $A_\yes$,
  while all the eigenvalues of $M'_x$ are at most~$\frac{1}{4}$
  if $x$ is in $A_\no$.

  Now consider the \textsc{Marriott-Watrous Amplification Procedure}
  associated with ${\bigl( V'_x, \Pi'_\init, \Pi'_\acc, \op{N}(\abs{x}), \op{N}(\abs{x}) \bigr)}$,
  which is exactly
  what the \textsc{Mild Amplification with Marriott-Watrous} associated with ${(V_x, p)}$ performs.
  By Proposition~\ref{Proposition: properties of Marriott-Watrous Amplification Procedure},
  the \textsc{Marriott-Watrous Amplification Procedure}
  associated with ${\bigl( V'_x, \Pi'_\init, \Pi'_\acc, \op{N}(\abs{x}), \op{N}(\abs{x}) \bigr)}$
  results in acceptance with probability at least
  \[
    1 - e^{- \frac{\op{N}(\abs{x})}{4}}
    \geq
    1 - e^{- \frac{\log \op{p}(\abs{x})}{\log e}}
    =
    1 - \frac{1}{\op{p}(\abs{x})}
  \]
  if $x$ is in $A_\yes$,
  and at most
  \[
    e^{- \frac{\op{N}(\abs{x})}{4}}
    \leq
    e^{- \frac{\log \op{p}(\abs{x})}{\log e}}
    =
    \frac{1}{\op{p}(\abs{x})}
  \]
  if $x$ is in $A_\no$,
  and the completeness and soundness follows.

  The \textsc{Marriott-Watrous Amplification Procedure}
  associated with ${\bigl( V'_x, \Pi'_\init, \Pi'_\acc, \op{N}(\abs{x}), \op{N}(\abs{x}) \bigr)}$
  uses extra workspace
  (relative to $V'_x$)
  of ${\op{\delta}_2 (\abs{x})}$~qubits
  for the function~$\function{\delta_2}{\Nonnegative}{\Natural}$ defined by
  ${\delta_2 = 2N + \ceil{\log (2N + 1)} + 1}$.
  As ${N = \bigceil{\frac{4 \log p}{\log e}}}$,
  $\delta_2$ is clearly logarithmic with respect to $p$, and thus, with respect to $\frac{p}{c-s}$ also.
  Hence, the \textsc{Mild Amplification with Marriott-Watrous} associated with ${(V_x, p)}$
  uses extra workspace
  (relative to $V_x$)
  of logarithmically many qubits with respect to $\frac{p}{c-s}$ also
  (which is determined by a function~${\delta = \delta_1 + \delta_2}$),
  as desired.
\end{proof}


\paragraph{Soundness error-reduction}

The rest of the construction is very similar to that in Subsection~\ref{Subsection: simple construction based on phase estimation}.

Again fix arbitrarily a function~$\function{p}{\Nonnegative}{\Natural}$
and functions~$\function{c,s}{\Nonnegative}{[0,1]}$ satisfying ${c > s}$,
and let $\function{N}{\Nonnegative}{\Natural}$ be a function defined by
\[
  N = \Bigceil{\frac{p}{2 \log (2p)}}.
\]

Fix an input~$x$.
Given the pair~${(V_x, p)}$,
consider the \textsc{Mild Amplification with Marriott-Watrous}
associated with ${(V_x, 4p^2)}$.
Let $V'_x$ be the unitary transformation induced by it,
let $\Pi'_\init$ be the projection onto the subspace spanned by the legal initial states of it,
and let $\Pi'_\acc$ be the projection onto the subspace spanned by the accepting states of it.
From the triplet~${\bigl( V'_x, \Pi'_\init, \Pi'_\acc \bigr)}$
and a positive integer~${\op{N}(\abs{x})}$,
one constructs the \textsc{AND-Type Repetition Procedure}
associated with ${\bigl( V'_x, \Pi'_\init, \Pi'_\acc, \op{N}(\abs{x}) \bigr)}$,
and performs it.
The resulting procedure is called
the \textsc{Soundness Error Reduction with Hybrid Construction},
and is summarized in Figure~\ref{Figure: Soundness Error Reduction with Hybrid Construction}.
\begin{figure}[t!]
  \begin{algorithm*}
        {
          \textsc{Soundness Error Reduction with Hybrid Construction}
          associated with $\boldsymbol{(V_x, p)}$
        }
    \begin{flushright}
      \begin{minipage}{0.9727\textwidth}
        Define a function~$\function{N}{\Nonnegative}{\Natural}$ by
        ${N = \bigceil{\frac{p}{2 \log (2p)}}}$.
        Consider the \textsc{Mild Amplification with Marriott-Watrous}
        associated with ${(V_x, 4p^2)}$.
        Let $V'_x$ be the unitary transformation induced by it,
        let $\Pi'_\init$ be the projection onto the subspace spanned by the legal initial states of it,
        and let $\Pi'_\acc$ be the projection onto the subspace spanned by the accepting states of it.\\
        Perform the \textsc{AND-Type Repetition Procedure}
        associated with ${\bigl( V'_x, \Pi'_\init, \Pi'_\acc, \op{N}(\abs{x}) \bigr)}$.
      \end{minipage}
    \end{flushright}
    \vspace{0.75\baselineskip}
  \end{algorithm*}
  \caption{
    The \textsc{Soundness Error Reduction with Hybrid Construction}.
  }
  \label{Figure: Soundness Error Reduction with Hybrid Construction}
\end{figure}

The following lemma is proved by using the \textsc{Soundness Error Reduction with Hybrid Construction}
combined with the properties of the \textsc{Mild Amplification with Marriott-Watrous}
used for proving Lemma~\ref{Lemma: mild amplification}.

\begin{lemma}
  For any functions~$\function{p, l_\regV, l_\regM}{\Nonnegative}{\Natural}$
  and any functions~$\function{c,s}{\Nonnegative}{[0,1]}$
  satisfying
  ${c > s}$,
  there exists a function~$\function{\delta}{\Nonnegative}{\Natural}$
  that is logarithmic with respect to ${\frac{p}{c - s}}$
  such that
  \[
    \unitaryQMASPACE[l_\regV, l_\regM](c, s)
    \subseteq
    \unitaryQMASPACE[l_\regV + \delta, l_\regM] \biggl( 1 - \frac{1}{p}, 2^{-2p} \biggr).
  \]
  \label{Lemma: soundness error reduction with hybrid construction}
\end{lemma}

\begin{proof}
  Let ${A = (A_\yes, A_\no)}$ be a problem in ${\unitaryQMASPACE[l_\regV, l_\regM](c, s)}$,
  and let ${V = \{ V_x \}_{x \in \Sigma^\ast}}$
  be the ${(l_\regV, l_\regM)}$-space-bounded quantum verifier witnessing this membership.
  Fix a function~$\function{p}{\Nonnegative}{\Natural}$ and an input~$x$ in $\Sigma^\ast$.
  The lemma is proved by considering
  the \textsc{Soundness Error Reduction with Hybrid Construction} associated with ${(V_x, p)}$.

  First consider the \textsc{Mild Amplification with Marriott-Watrous}
  associated with ${(V_x, 4p^2)}$.
  Let $V'_x$ be the unitary transformation induced by it,
  let $\Pi'_\init$ be the projection onto the subspace spanned by the legal initial states of it,
  and let $\Pi'_\acc$ be the projection onto the subspace spanned by the accepting states of it.

  Lemma~\ref{Lemma: mild amplification}
  and its proof based on the \textsc{Mild Amplification with Marriott-Watrous}
  ensure that
  $A$ is in
  ${\unitaryQMASPACE[l_\regV + \delta_1, l_\regM] \bigl( 1 - \frac{1}{4p^2}, \frac{1}{4p^2} \bigr)}$
  for some function~$\function{\delta_1}{\Nonnegative}{\Natural}$
  that is logarithmic with respect to ${\frac{p}{c - s}}$,
  and this inclusion is certified by
  the \textsc{Mild Amplification with Marriott-Watrous}
  associated with ${(V_x, 4p^2)}$.
  This in particular implies that
  the Hermitian operator~%
  ${M'_x = \Pi'_\init \conjugate{(V'_x)} \Pi'_\acc V'_x \Pi'_\init}$
  has an eigenvalue at least~%
  ${1 - \frac{1}{4 (\op{p}(\abs{x}))^2}}$
  if $x$ is in $A_\yes$,
  while all the eigenvalues of $M'_x$ are at most~%
  $\frac{1}{4 (\op{p}(\abs{x}))^2}$
  if $x$ is in $A_\no$.

  Now consider the \textsc{AND-Type Repetition Procedure}
  associated with ${\bigl( V'_x, \Pi'_\init, \Pi'_\acc, \op{N}(\abs{x}) \bigr)}$,
  which is exactly
  what the \textsc{Soundness Error Reduction with Hybrid Construction} associated with ${(V_x, p)}$ performs.
  By Proposition~\ref{Proposition: properties of AND-Type Repetition Procedure},
  the \textsc{AND-Type Repetition Procedure}
  associated with ${\bigl( V'_x, \Pi'_\init, \Pi'_\acc, \op{N}(\abs{x}) \bigr)}$
  results in acceptance with probability at least
  \[
    \biggl( 1 - \frac{1}{4 (\op{p}(\abs{x}))^2} \biggr)^{2 \op{N}(\abs{x})}
    >
    \biggl( 1 - \frac{1}{4 (\op{p}(\abs{x}))^2} \biggr)^{\frac{\op{p}(\abs{x})}{\log (2 \op{p}(\abs{x}))} + 2}
    >
    1 - \frac{1}{\op{p}(\abs{x})}
  \]
  if $x$ is in $A_\yes$,
  and at most
  \[
    \biggl( \frac{1}{4 (\op{p}(\abs{x}))^2} \biggr)^{2 \op{N}(\abs{x})}
    \leq
    \Biggl[
      \biggl( \frac{1}{2 \op{p}(\abs{x})} \biggr)^{\frac{\op{p}(\abs{x})}{\log (2 \op{p}(\abs{x}))}}
    \Biggr]^2
    =
    2^{- 2 \op{p}(\abs{x})}
  \]
  if $x$ is in $A_\no$,
  and the completeness and soundness follows.

  The \textsc{AND-Type Repetition Procedure}
  associated with
  ${\bigl( V'_x, \Pi'_\init, \Pi'_\acc, \op{N}(\abs{x}) \bigr)}$
  uses extra workspace
  (relative to $V'_x$)
  of ${\op{\delta}_2 (\abs{x})}$~qubits
  for the function~$\function{\delta_2}{\Nonnegative}{\Natural}$ defined by
  ${\delta_2 = \ceil{\log (2N + 1)}}$.
  As ${N = \bigceil{\frac{p}{2 \log (2p)}}}$,
  $\delta_2$ is clearly logarithmic with respect to $p$, and thus, with respect to $\frac{p}{c-s}$ also.
  Hence, the \textsc{Soundness Error Reduction with Hybrid Construction} associated with ${(V_x, p)}$
  uses extra workspace
  (relative to $V_x$)
  of logarithmically many qubits with respect to $\frac{p}{c-s}$ also
  (which is determined by a function~${\delta = \delta_1 + \delta_2}$),
  as desired.
\end{proof}


\paragraph{Space-efficient error reduction based on hybrid construction}

Again fix arbitrarily a function~$\function{p}{\Nonnegative}{\Natural}$
and functions~$\function{c,s}{\Nonnegative}{[0,1]}$ satisfying ${c > s}$,
and let $\function{N}{\Nonnegative}{\Natural}$ be a function defined by
\[
  N = \Bigceil{\frac{p}{2 \log p}}.
\]

Fix an input~$x$.
Given the pair~${(V_x, p)}$,
consider the \textsc{Soundness Error Reduction with Hybrid Construction}
associated with ${(V_x, p)}$.
Let $V'_x$ be the unitary transformation induced by it,
let $\Pi'_\init$ be the projection onto the subspace spanned by the legal initial states of it,
and let $\Pi'_\acc$ be the projection onto the subspace spanned by the accepting states of it.
From the triplet~${\bigl( V'_x, \Pi'_\init, \Pi'_\acc \bigr)}$
and a positive integer~${\op{N}(\abs{x})}$,
one constructs the \textsc{OR-Type Repetition Procedure}
associated with ${\bigl( V'_x, \Pi'_\init, \Pi'_\acc, \op{N}(\abs{x}) \bigr)}$,
and performs it.
The resulting procedure is called
the \textsc{Space-Efficient Error Reduction Based on Hybrid Construction},
and is summarized in Figure~\ref{Figure: Space-Efficient Error Reduction Based on Hybrid Construction}.
\begin{figure}[t!]
  \begin{algorithm*}
        {
          \textsc{Space-Efficient Error Reduction Based on Hybrid Construction}
          associated with $\boldsymbol{(V_x, p)}$
        }
    \begin{flushright}
      \begin{minipage}{0.9727\textwidth}
        Define a function~$\function{N}{\Nonnegative}{\Natural}$ by
        ${N = \bigceil{\frac{p}{2 \log p}}}$.
        Consider the \textsc{Soundness Error Reduction with Hybrid Construction}
        associated with ${(V_x, p)}$.
        Let $V'_x$ be the unitary transformation induced by it,
        let $\Pi'_\init$ be the projection onto the subspace spanned by the legal initial states of it,
        and let $\Pi'_\acc$ be the projection onto the subspace spanned by the accepting states of it.\\
        Perform the \textsc{OR-Type Repetition Procedure}
        associated with ${\bigl( V'_x, \Pi'_\init, \Pi'_\acc, \op{N}(\abs{x}) \bigr)}$.
      \end{minipage}
    \end{flushright}
    \vspace{0.75\baselineskip}
  \end{algorithm*}
  \caption{
    The \textsc{Space-Efficient Error Reduction Based on Hybrid Construction}.
  }
  \label{Figure: Space-Efficient Error Reduction Based on Hybrid Construction}
\end{figure}

Now Theorem~\ref{Theorem: main theorem},
the main theorem of this paper,
is ready to be proved
by using the \textsc{Space-Efficient Error Reduction Based on Hybrid Construction}
combined with the properties of the \textsc{Soundness Error Reduction with Hybrid Construction}
used for proving Lemma~\ref{Lemma: soundness error reduction with hybrid construction}.

\begin{proof}
      [%
        Proof of Theorem~\ref{Theorem: main theorem}
        (via the hybrid construction of phase estimation and Marriott-Watrous)
      ]
  Let ${A\hspace{-0.167mm} =\hspace{-0.167mm} (A_\yes,\hspace{-0.166mm} A_\no)}$
  be a problem in ${\unitaryQMASPACE[l_\regV, l_\regM](c, s)}$,
  and let ${V = \{ V_x \}_{x \in \Sigma^\ast}}$
  be the ${(l_\regV, l_\regM)}$-space-bounded quantum verifier witnessing this membership.
  Fix a function~$\function{p}{\Nonnegative}{\Natural}$ and an input~$x$ in $\Sigma^\ast$.
  The theorem is proved by considering
  the \textsc{Space-Efficient Error Reduction Based on Hybrid Construction} associated with ${(V_x, p)}$.

  First consider the \textsc{Soundness Error Reduction with Hybrid Construction}
  associated with ${(V_x, p)}$.
  Let $V'_x$ be the unitary transformation induced by it,
  let $\Pi'_\init$ be the projection onto the subspace spanned by the legal initial states of it,
  and let $\Pi'_\acc$ be the projection onto the subspace spanned by the accepting states of it.

  Lemma~\ref{Lemma: soundness error reduction with hybrid construction} and its proof
  ensure that
  $A$ is in
  ${\unitaryQMASPACE[l_\regV + \delta_1, l_\regM] \bigl( 1 - \frac{1}{p}, 2^{-2p} \bigr)}$
  for some function~$\function{\delta_1}{\Nonnegative}{\Natural}$
  that is logarithmic with respect to ${\frac{p}{c - s}}$,
  and this inclusion is certified by
  the \textsc{Soundness Error Reduction with Hybrid Construction}
  associated with ${(V_x, p)}$.
  This in particular implies that
  the Hermitian operator~%
  ${M'_x = \Pi'_\init \conjugate{(V'_x)} \Pi'_\acc V'_x \Pi'_\init}$
  has an eigenvalue at least~${1 - \frac{1}{\op{p}(\abs{x})}}$
  if $x$ is in $A_\yes$,
  while all the eigenvalues of $M'_x$ are at most~$2^{- 2 \op{p}(\abs{x})}$
  if $x$ is in $A_\no$.

  Now consider the \textsc{OR-Type Repetition Procedure}
  associated with ${\bigl( V'_x, \Pi'_\init, \Pi'_\acc, \op{N}(\abs{x}) \bigr)}$,
  which is exactly what
  the \textsc{Space-Efficient Error Reduction Based on Hybrid Construction} associated with ${(V_x, p)}$
  performs.
  By Proposition~\ref{Proposition: properties of OR-Type Repetition Procedure},
  the \textsc{OR-Type Repetition Procedure}
  associated with ${\bigl( V'_x, \Pi'_\init, \Pi'_\acc, \op{N}(\abs{x}) \bigr)}$
  results in acceptance with probability at least
  \[
    1 - \biggl( \frac{1}{\op{p}(\abs{x})} \biggr)^{2 \op{N}(\abs{x})}
    \geq
    1 - \biggl( 2^{- \log \op{p}(\abs{x})} \biggr)^{\frac{\op{p}(\abs{x})}{\log \op{p}(\abs{x})}}
    =
    1 - 2^{- \op{p}(\abs{x})}
  \]
  if $x$ is in $A_\yes$,
  and at most
  \[
    1 - \Bigl( 1 - 2^{- 2 \op{p}(\abs{x})} \Bigr)^{2 \op{N}(\abs{x})}
    <
    1 - \Bigl( 1 - 2^{- 2 \op{p}(\abs{x})} \Bigr)^{\frac{\op{p}(\abs{x})}{\log \op{p}(\abs{x})} + 2}
    <
    2^{- \op{p}(\abs{x})}
  \]
  if $x$ is in $A_\no$,
  and the completeness and soundness follows.

  The \textsc{OR-Type Repetition Procedure}
  associated with
  ${\bigl( V'_x, \Pi'_\init, \Pi'_\acc, \op{N}(\abs{x}) \bigr)}$
  uses extra workspace
  (relative to $V'_x$)
  of ${\op{\delta}_2 (\abs{x})}$~qubits
  for the function~$\function{\delta_2}{\Nonnegative}{\Natural}$ defined by
  ${\delta_2 = \ceil{\log (2N + 1)}}$.
  As ${N = \bigceil{\frac{p}{2 \log p}}}$,
  $\delta_2$ is clearly logarithmic with respect to $p$, and thus, with respect to $\frac{p}{c-s}$ also.
  Hence,
  the \textsc{Space-Efficient Error Reduction Based on Hybrid Construction} associated with ${(V_x, p)}$
  uses extra workspace
  (relative to $V_x$)
  of logarithmically many qubits with respect to $\frac{p}{c-s}$ also
  (which is determined by a function~${\delta = \delta_1 + \delta_2}$),
  as desired.
\end{proof}

A straightforward calculation shows that
the \textsc{Space-Efficient Error Reduction Based on Hybrid Construction} associated with ${(V_x, p)}$
uses
${\op{O} \bigl( \frac{1}{c - s} \cdot \frac{p^2}{\log p} \bigr)}$~calls
of $V_x$ and its inverse.


\subsection{Exactly implementable construction based on random guess}
\label{Subsection: exactly implementable construction based on random guess}

The third proof is via the exactly implementable construction based on random guess.


\paragraph{Mild completeness amplification with a guess}

Fix a function~$\function{p}{\Nonnegative}{\Natural}$
and functions~$\function{c,s}{\Nonnegative}{[0,1]}$ satisfying ${c > s}$ arbitrarily,
and let $\function{l, C}{\Nonnegative}{\Natural}$ be functions defined by
\[
  l = \Bigceil{\frac{1}{2} \log \frac{p}{(c - s)^2}},
  \quad
  C = \ceil{2^l c}.
\]

Fix an input~$x$ and a positive integer~$k$ in ${\{1, \dotsc, 2^{\op{l}(\abs{x})}\}}$.
Given the triplet~${(V_x, \Pi_\init, \Pi_\acc)}$
and the integer~$k$,
one first constructs the \textsc{Additive Adjustment Procedure}
associated with ${\bigl( V_x, \Pi_\init, \Pi_\acc, \op{l}(\abs{x}), k \bigr)}$,
if $k$ is at least ${\op{C}(\abs{x})}$
(and automatically rejects otherwise
so that no $k$ can result in a good guess at the acceptance probability
when the actual value of it is unallowably small).
Let $V'_{x,k}$ be the unitary transformation induced by it,
let $\Pi'_\init$ be the projection onto the subspace spanned by the legal initial states of it,
and let $\Pi'_{\acc, k}$ be the projection onto the subspace spanned by the accepting states of it.
Next,
from the triplet~${\bigl( V'_{x,k}, \Pi'_\init, \Pi'_{\acc, k} \bigr)}$,
one constructs the \textsc{Reflection Procedure}
associated with ${\bigl( V'_{x,k}, \Pi'_\init, \Pi'_{\acc, k} \bigr)}$,
and performs it.
The resulting procedure is called the \textsc{Mild Completeness Amplification with Guess~$k$},
and is summarized as in Figure~\ref{Figure: Mild Completeness Amplification with Guess}.

\begin{figure}[t!]
  \begin{algorithm*}
        {
          \textsc{Mild Completeness Amplification with Guess~$\boldsymbol{k}$}
          associated with $\boldsymbol{(V_x, p)}$
        }
    \begin{flushright}
      \begin{minipage}{0.9727\textwidth}
        Define functions~$l$~and~$C$ by
        ${l = \bigceil{\frac{1}{2} \log \frac{p}{(c - s)^2}}}$
        and
        ${C = \ceil{2^l c}}$.
        Let $\Pi_\init$ and $\Pi_\acc$ be the projections onto the subspaces
        spanned by the legal initial states and the accepting states, respectively,
        in the verification with $V_x$.
        Given an integer~$k$ in ${\{1, \dotsc, 2^{\op{l}(\abs{x})}\}}$ as a guess,
        consider the \textsc{Additive Adjustment Procedure}
        associated with ${(V_x, \Pi_\init, \Pi_\acc, \op{l}(\abs{x}), k)}$.
        Let $V'_{x,k}$ be the unitary transformation induced by it,
        let $\Pi'_\init$ be the projection onto the subspace spanned by the legal initial states of it,
        and let $\Pi'_{\acc, k}$ be the projection onto the subspace spanned by the accepting states of it.\\
        Reject if ${k < \op{C}(\abs{x})}$,
        and continue otherwise by performing the \textsc{Reflection Procedure}
        associated with ${\bigl( V'_{x,k}, \Pi'_\init, \Pi'_{\acc, k} \bigr)}$.
      \end{minipage}
    \end{flushright}
    \vspace{0.75\baselineskip}
  \end{algorithm*}
  \caption{
    The \textsc{Mild Completeness Amplification with Guess~$k$}.
  }
  \label{Figure: Mild Completeness Amplification with Guess}
\end{figure}

From the properties of the \textsc{Additive Adjustment Procedure} and the \textsc{Reflection Procedure}
(Propositions~\ref{Proposition: properties of Additive Adjustment Procedure}~and~\ref{Proposition: properties of Reflection Procedure}),
one can show the following lemma.

\begin{lemma}
  Given functions~$\function{l_\regV, l_\regM}{\Nonnegative}{\Natural}$
  and $\function{c,s}{\Nonnegative}{[0,1]}$ satisfying ${c > s}$,
  let ${A = (A_\yes, A_\no)}$ be a problem in ${\unitaryQMASPACE[l_\regV, l_\regM](c, s)}$,
  and let ${V = \{ V_x \}_{x \in \Sigma^\ast}}$
  be the ${(l_\regV, l_\regM)}$-space-bounded quantum verifier witnessing this membership.
  Then, for any function~$\function{p}{\Nonnegative}{\Natural}$
  and for every $x$ in $\Sigma^\ast$,
  letting ${l = \bigceil{\frac{1}{2} \log \frac{p}{(c - s)^2}}}$,
  the following properties hold:
  \begin{description}
  \item[\textnormal{(Completeness)}]
    If $x$ is in $A_\yes$,
    there exists an integer~$k$ in ${\{1, \dotsc, 2^{\op{l}(\abs{x})}\}}$ as a guess
    such that
    the \textsc{Mild Completeness Amplification with Guess~$k$} associated with ${(V_x, p)}$
    results in acceptance with probability at least
    ${1 - \frac{\left( \op{c}(\abs{x}) - \op{s}(\abs{x}) \right)^2}{\op{p}(\abs{x})}}$.
  \item[\textnormal{(Soundness)}]
    If $x$ is in $A_\no$,
    for any integer~$k$ in ${\{1, \dotsc, 2^{\op{l}(\abs{x})}\}}$ as a guess,
    the \textsc{Mild Completeness Amplification with Guess~$k$} associated with ${(V_x, p)}$
    results in acceptance with probability at most
    ${1 - \bigl( \op{c}(\abs{x}) - \op{s}(\abs{x}) \bigr)^2}$.
  \end{description}
  \label{Lemma: mild completeness amplification with a guess}
\end{lemma}

\begin{proof}
  Let $\function{C}{\Nonnegative}{\Natural}$ be a function defined by ${C = \ceil{2^l c}}$,
  and let $\Pi_\init$ and $\Pi_\acc$ be the projections onto the subspaces
  spanned by the legal initial states and the accepting states, respectively,
  in the verification with $V_x$.
  For the \textsc{Additive Adjustment Procedure}
  associated with ${(V_x, \Pi_\init, \Pi_\acc, \op{l}(\abs{x}), k)}$,
  let $V'_{x,k}$ be the unitary transformation induced by it,
  let $\Pi'_\init$ be the projection onto the subspace spanned by the legal initial states of it,
  and let $\Pi'_{\acc, k}$ be the projection onto the subspace spanned by the accepting states of it.

  First suppose that $x$ is in $A_\yes$.
  The Hermitian operator~${M_x = \Pi_\init \conjugate{V_x} \Pi_\acc V_x \Pi_\init}$
  in this case
  has an eigenvalue~$\lambda_x$ that is at least ${\op{c}(\abs{x})}$.

  Fix ${k = \bigceil{2^{\op{l}(\abs{x})} \lambda_x}}$
  in ${\{\op{C}(\abs{x}), \dotsc, 2^{\op{l}(\abs{x})}\}}$.

  By Proposition~\ref{Proposition: properties of Additive Adjustment Procedure},
  the Hermitian operator~%
  ${
    M'_{x,k} = \Pi'_\init \conjugate{\bigl( V'_{x,k} \bigr)} \Pi'_{\acc, k} V'_{x,k} \Pi'_\init
  }$
  must have an eigenvalue
  \[
    \lambda'_{x,k}
    =
    \frac{1}{2} - \frac{1}{2} \biggl( \frac{k}{2^{\op{l}(\abs{x})}} - \lambda_x \biggr),
  \]
  which must satisfy that
  \[
    \frac{1}{2} - \frac{\op{c}(\abs{x}) - \op{s}(\abs{x})}{2 \sqrt{\op{p}(\abs{x})}}
    \leq
    \frac{1}{2} - 2^{- (\op{l}(\abs{x}) + 1)}
    <
    \lambda'_{x,k}
    \leq
    \frac{1}{2}
  \]
  for ${k = \bigceil{2^{\op{l}(\abs{x})} \lambda_x}}$
  in ${\{\op{C}(\abs{x}), \dotsc, 2^{\op{l}(\abs{x})}\}}$.

  Hence,
  by Proposition~\ref{Proposition: properties of Reflection Procedure},
  the \textsc{Reflection Procedure}
  associated with ${\bigl( V'_{x,k}, \Pi'_\init, \Pi'_{\acc, k} \bigr)}$
  results in acceptance with probability at least
  \[
    1 - \biggl( \frac{k}{2^{\op{l}(\abs{x})}} - \lambda_x \biggr)^2
    >
    1 - 2^{- 2 \op{l}(\abs{x})}
    \geq
    1 - \frac{\bigl( \op{c}(\abs{x}) - \op{s}(\abs{x}) \bigr)^2}{\op{p}(\abs{x})},
  \]
  which proves the completeness.

  Now suppose that $x$ is in $A_\no$,
  which implies that
  all the eigenvalues of $M_x$ are at most ${\op{s}(\abs{x})}$.
  It follows from Proposition~\ref{Proposition: properties of Additive Adjustment Procedure} that,
  for any $k$ in ${\{\op{C}(\abs{x}), \dotsc, 2^{\op{l}(\abs{x})}\}}$,
  all the eigenvalues of $M'_{x,k}$ are at most
  \[
    \frac{1}{2} - \frac{1}{2} \biggl( \frac{k}{2^{\op{l}(\abs{x})}} - \op{s}(\abs{x}) \biggr)
    \leq
    \frac{1}{2} - \frac{1}{2} \biggl( \frac{\op{C}(\abs{x})}{2^{\op{l}(\abs{x})}} - \op{s}(\abs{x}) \biggr)
    \leq
    \frac{1}{2} - \frac{1}{2} \bigl( \op{c}(\abs{x}) - \op{s}(\abs{x}) \bigr).
  \]
  Therefore,
  Proposition~\ref{Proposition: properties of Reflection Procedure} ensures that,
  for any $k$ in ${\{\op{C}(\abs{x}), \dotsc, 2^{\op{l}(\abs{x})}\}}$,
  the \textsc{Reflection Procedure}
  associated with ${\bigl( V'_{x,k}, \Pi'_\init, \Pi'_{\acc, k} \bigr)}$
  results in acceptance with probability at most
  \[
    1 - \bigl( \op{c}(\abs{x}) - \op{s}(\abs{x}) \bigr)^2.
  \]
  As it always rejects when $k$ is less than ${\op{C}(\abs{x})}$,
  the \textsc{Mild Completeness Amplification with Guess~$k$} associated with ${(V_x, p)}$
  results in acceptance with probability at most
  ${
    1 - \bigl( \op{c}(\abs{x}) - \op{s}(\abs{x}) \bigr)^2
  }$
  for any $k$ in ${\{1, \dotsc, 2^{\op{l}(\abs{x})}\}}$,
  and the soundness follows.
\end{proof}


\paragraph{Soundness error reduction with a guess}

Again fix a function~$\function{p}{\Nonnegative}{\Natural}$
and functions~$\function{c,s}{\Nonnegative}{[0,1]}$ satisfying ${c > s}$,
arbitrarily.
Let $\function{l, N}{\Nonnegative}{\Natural}$ be functions defined by
\[
  l = \biggceil{\frac{1}{2} \log \frac{6p}{(c - s)^2}},
  \quad
  N = \biggceil{\frac{p}{2 (c - s)^2}}.
\]

Fix an input~$x$ and an integer~$k$ in ${\{1, \dotsc, 2^{\op{l}(\abs{x})}\}}$.
Given the pair~${(V_x, p)}$
and the integer~$k$,
consider the \textsc{Mild Completeness Amplification with Guess~$k$}
associated with ${(V_x, 6p)}$.
As before, let $V'_{x,k}$ be the unitary transformation induced by it,
let $\Pi'_\init$ be the projection onto the subspace spanned by the legal initial states of it,
and let $\Pi'_{\acc, k}$ be the projection onto the subspace spanned by the accepting states of it.
From the triplet~${\bigl( V'_{x,k}, \Pi'_\init, \Pi'_{\acc, k} \bigr)}$
and a positive integer~${\op{N}(\abs{x})}$,
one constructs the \textsc{AND-Type Repetition Procedure}
associated with ${\bigl( V'_{x,k}, \Pi'_\init, \Pi'_{\acc, k}, \op{N}(\abs{x}) \bigr)}$,
and performs it.
The resulting procedure is called the \textsc{Soundness Error Reduction with Guess~$k$},
and is summarized in Figure~\ref{Figure: Soundness Error Reduction with Guess}.

\begin{figure}[t!]
  \begin{algorithm*}
        {
          \textsc{Soundness Error Reduction with Guess~$\boldsymbol{k}$}
          associated with $\boldsymbol{(V_x, p)}$
        }
    \begin{flushright}
      \begin{minipage}{0.9727\textwidth}
        Define functions~$l$~and~$N$ by
        ${l = \bigceil{\frac{1}{2} \log \frac{6p}{(c - s)^2}}}$
        and
        ${N = \bigceil{\frac{p}{2 (c - s)^2}}}$.
        Given an integer~$k$ in ${\{1, \dotsc, 2^{\op{l}(\abs{x})}\}}$,
        consider the \textsc{Mild Completeness Amplification with Guess~$k$}
        associated with ${(V_x, 6p)}$.
        Let $V'_{x,k}$ be the unitary transformation induced by it,
        let $\Pi'_\init$ be the projection onto the subspace spanned by the legal initial states of it,
        and let $\Pi'_{\acc, k}$ be the projection onto the subspace spanned by the accepting states of it.\\
        Perform the \textsc{AND-Type Repetition Procedure}
        associated with ${\bigl( V'_{x,k}, \Pi'_\init, \Pi'_{\acc, k}, \op{N}(\abs{x}) \bigr)}$.
      \end{minipage}
    \end{flushright}
    \vspace{0.75\baselineskip}
  \end{algorithm*}
  \caption{
    The \textsc{Soundness Error Reduction with Guess~$k$}.
  }
  \label{Figure: Soundness Error Reduction with Guess}
\end{figure}

From the properties of the \textsc{AND-Type Repetition Procedure}
and the \textsc{Mild Completeness Amplification with Guess~$k$}
(Proposition~\ref{Proposition: properties of AND-Type Repetition Procedure}
and Lemma~\ref{Lemma: mild completeness amplification with a guess}),
one can show the following lemma.

\begin{lemma}
  Given functions~$\function{l_\regV, l_\regM}{\Nonnegative}{\Natural}$
  and $\function{c,s}{\Nonnegative}{[0,1]}$ satisfying ${c > s}$,
  let ${A = (A_\yes, A_\no)}$ be a problem in ${\unitaryQMASPACE[l_\regV, l_\regM](c, s)}$,
  and let ${V = \{ V_x \}_{x \in \Sigma^\ast}}$
  be the ${(l_\regV, l_\regM)}$-space-bounded quantum verifier witnessing this membership.
  Then, for any function~$\function{p}{\Nonnegative}{\Natural}$
  and for every $x$ in $\Sigma^\ast$,
  letting ${l = \bigceil{\frac{1}{2} \log \frac{6p}{(c - s)^2}}}$,
  the following properties hold:
  \begin{description}
  \item[\textnormal{(Completeness)}]
    If $x$ is in $A_\yes$,
    there exists an integer~$k$ in ${\{1, \dotsc, 2^{\op{l}(\abs{x})}\}}$ as a guess
    such that
    the \textsc{Soundness Error-Reduction with Guess~$k$} associated with ${(V_x, p)}$
    results in acceptance with probability at least~$\frac{1}{2}$.
  \item[\textnormal{(Soundness)}]
    If $x$ is in $A_\no$,
    for any integer~$k$ in ${\{1, \dotsc, 2^{\op{l}(\abs{x})}\}}$ as a guess,
    the \textsc{Soundness Error-Reduction with Guess~$k$} associated with ${(V_x, p)}$
    results in acceptance with probability at most~${2^{- \op{p}(\abs{x})}}$.
  \end{description}
  \label{Lemma: soundness error reduction with a guess}
\end{lemma}

\begin{proof}
  Let $\function{C}{\Nonnegative}{\Natural}$ be a function defined by ${C = \ceil{2^l c}}$,
  and let $\Pi_\init$ and $\Pi_\acc$ be the projections onto the subspaces
  spanned by the legal initial states and the accepting states, respectively,
  in the verification with $V_x$.
  For the \textsc{Mild Completeness Amplification with Guess~$k$}
  associated with ${(V_x, 6p)}$,
  let $V'_{x,k}$ be the unitary transformation induced by it,
  let $\Pi'_\init$ be the projection onto the subspace spanned by the legal initial states of it,
  and let $\Pi'_{\acc, k}$ be the projection onto the subspace spanned by the accepting states of it.
  Then, for a function~$\function{N}{\Nonnegative}{\Natural}$ defined by
  ${N = \bigceil{\frac{p}{2 (c - s)^2}}}$
  and for the \textsc{AND-Type Repetition Procedure}
  associated with ${\bigl( V'_{x,k}, \Pi'_\init, \Pi'_{\acc, k}, \op{N}(\abs{x}) \bigr)}$,
  let $V''_{x,k}$ be the unitary transformation induced by it,
  let $\Pi''_\init$ be the projection onto the subspace spanned by the legal initial states of it,
  and let $\Pi''_{\acc, k}$ be the projection onto the subspace spanned by the accepting states of it.

  First suppose that $x$ is in $A_\yes$.
  The Hermitian operator~${M_x = \Pi_\init \conjugate{V_x} \Pi_\acc V_x \Pi_\init}$
  in this case
  has an eigenvalue~$\lambda_x$ that is at least ${\op{c}(\abs{x})}$.

  Fix ${k = \bigceil{2^{\op{l}(\abs{x})} \lambda_x}}$
  in ${\{\op{C}(\abs{x}), \dotsc, 2^{\op{l}(\abs{x})}\}}$.

  By Lemma~\ref{Lemma: mild completeness amplification with a guess},
  the Hermitian operator~%
  ${
    M'_{x,k} = \Pi'_\init \conjugate{\bigl( V'_{x,k} \bigr)} \Pi'_{\acc, k} V'_{x,k} \Pi'_\init
  }$
  must have an eigenvalue~%
  \[
    \lambda'_{x,k}
    >
    1 - \frac{\bigl( \op{c}(\abs{x}) - \op{s}(\abs{x}) \bigr)^2}{6 \op{p}(\abs{x})}
  \]
  for ${k = \bigceil{2^{\op{l}(\abs{x})} \lambda_x}}$
  in ${\{\op{C}(\abs{x}), \dotsc, 2^{\op{l}(\abs{x})}\}}$.
  Hence,
  by Proposition~\ref{Proposition: properties of AND-Type Repetition Procedure},
  the \textsc{AND-Type Repetition Procedure}
  associated with ${\bigl( V'_{x,k}, \Pi'_\init, \Pi'_{\acc, k}, \op{N}(\abs{x}) \bigr)}$
  results in acceptance with probability at least
  \[
    \Biggl[
      1 - \frac{\bigl( \op{c}(\abs{x}) - \op{s}(\abs{x}) \bigr)^2}{6 \op{p}(\abs{x})}
    \Biggr]^{2 \op{N}(\abs{x})}
    \geq
    \Biggl[
      1 - \frac{\bigl( \op{c}(\abs{x}) - \op{s}(\abs{x}) \bigr)^2}{6 \op{p}(\abs{x})}
    \Biggr]^{\frac{\op{p}(\abs{x})}{\left( \op{c}(\abs{x}) - \op{s}(\abs{x}) \right)^2} + 2}
    >
    \frac{1}{2},
  \]
  which proves the completeness.

  Now suppose that $x$ is in $A_\no$,
  which implies that
  all the eigenvalues of $M_x$ are at most ${\op{s}(\abs{x})}$.
  It follows from Lemma~\ref{Lemma: mild completeness amplification with a guess} that,
  for any $k$ in ${\{\op{C}(\abs{x}), \dotsc, 2^{\op{l}(\abs{x})}\}}$,
  all the eigenvalues of $M'_{x,k}$ are at most
  \[
    1 - \bigl( \op{c}(\abs{x}) - \op{s}(\abs{x}) \bigr)^2.
  \]
  From Proposition~\ref{Proposition: properties of AND-Type Repetition Procedure},
  this implies that,
  for any $k$ in ${\{\op{C}(\abs{x}), \dotsc, 2^{\op{l}(\abs{x})}\}}$,
  the \textsc{AND-Type Repetition Procedure}
  associated with ${\bigl( V'_{x,k}, \Pi'_\init, \Pi'_{\acc, k}, \op{N}(\abs{x}) \bigr)}$
  results in acceptance with probability at most
  \[
    \Bigl[ 1 - \bigl( \op{c}(\abs{x}) - \op{s}(\abs{x}) \bigr)^2 \Bigr]^{2 \op{N}(\abs{x})}
    \leq
    \Bigl[
      1 - \bigl( \op{c}(\abs{x}) - \op{s}(\abs{x}) \bigr)^2
    \Bigr]^{\frac{\op{p}(\abs{x})}{\left( \op{c}(\abs{x}) - \op{s}(\abs{x}) \right)^2}}
    <
    e^{- \op{p}(\abs{x})}
    <
    2^{- \op{p}(\abs{x})}.
  \]
  As it always rejects when $k$ is less than ${\op{C}(\abs{x})}$,
  the \textsc{Soundness Error Reduction with Guess~$k$} associated with ${(V_x, p)}$
  results in acceptance with probability at most ${2^{- \op{p}(\abs{x})}}$
  for any $k$ in ${\{1, \dotsc, 2^{\op{l}(\abs{x})}\}}$,
  and the soundness follows.
\end{proof}


\paragraph{Soundness error reduction with a random guess}

Again fix arbitrarily a function~$\function{p}{\Nonnegative}{\Natural}$
and functions~$\function{c,s}{\Nonnegative}{[0,1]}$ satisfying ${c > s}$,
and let $\function{l}{\Nonnegative}{\Natural}$ be a function defined by
\[
  l = \biggceil{\frac{1}{2} \log \frac{6p}{(c - s)^2}}.
\]

Fix an input~$x$.
Given the pair~${(V_x, p)}$,
consider choosing an integer~$k$ from ${\{1, \dotsc, 2^{\op{l}(\abs{x})}\}}$
uniformly at random,
and then performing the \textsc{Soundness Error Reduction with Guess~$k$} associated with ${(V_x, p)}$.
The resulting procedure is called the \textsc{Soundness Error Reduction with Random Guess}
and is summarized in Figure~\ref{Figure: Soundness Error Reduction with Random Guess}.
\begin{figure}[t!]
  \begin{algorithm*}
        {
          \textsc{Soundness Error Reduction with Random Guess}
          associated with $\boldsymbol{(V_x, p)}$
        }
    \begin{flushright}
      \begin{minipage}{0.9727\textwidth}
        Define a function~$l$ by
        ${l = \bigceil{\frac{1}{2} \log \frac{6p}{(c - s)^2}}}$.\\
        Pick an integer~$k$ from ${\{1, \dotsc, 2^{\op{l}(\abs{x})}\}}$ uniformly at random
        and perform the \textsc{Soundness Error Reduction with Guess~$k$} associated with ${(V_x, p)}$.
      \end{minipage}
    \end{flushright}
  \vspace{0.75\baselineskip}
  \end{algorithm*}
  \caption{
    The \textsc{Soundness Error Reduction with Random Guess}.
  }
  \label{Figure: Soundness Error Reduction with Random Guess}
\end{figure}

The following lemma is proved by using the \textsc{Soundness Error Reduction with Random Guess}
combined with the properties of the \textsc{Soundness Error Reduction with Guess~$k$}
stated in Lemma~\ref{Lemma: soundness error reduction with a guess}.

\begin{lemma}
  For any functions~$\function{p, l_\regV, l_\regM}{\Nonnegative}{\Natural}$
  and any functions~$\function{c,s}{\Nonnegative}{[0,1]}$
  satisfying
  ${c > s}$
  and
  ${\frac{c - s}{4 \sqrt{6p}} > 2^{-p}}$
  (which in particular holds when ${p > 2 \log \frac{4 \sqrt{3}}{c - s}}$),
  there exists a function~$\function{\delta}{\Nonnegative}{\Natural}$
  that is logarithmic with respect to ${\frac{p}{c - s}}$
  such that
  \[
    \unitaryQMASPACE[l_\regV, l_\regM](c, s)
    \subseteq
    \unitaryQMASPACE[l_\regV + \delta, l_\regM] \biggl( \frac{c - s}{4 \sqrt{6p}}, 2^{-p} \biggr).
  \]
  \vspace{-0.75\baselineskip}
  \label{Lemma: soundness error reduction with a random guess}
\end{lemma}

\begin{proof}
  Let ${A = (A_\yes, A_\no)}$ be a problem in ${\unitaryQMASPACE[l_\regV, l_\regM](c, s)}$,
  and let ${V = \{ V_x \}_{x \in \Sigma^\ast}}$
  be the ${(l_\regV, l_\regM)}$-space-bounded quantum verifier witnessing this membership.
  Fix a function~$\function{p}{\Nonnegative}{\Natural}$ satisfying ${\frac{c - s}{4 \sqrt{6p}} > 2^{-p}}$
  and an input~$x$ in $\Sigma^\ast$.
  The lemma is proved by considering
  the \textsc{Soundness Error Reduction with Random Guess} associated with ${(V_x, p)}$.

  Lemma~\ref{Lemma: soundness error reduction with a guess} ensures that,
  if $x$ is in $A_\yes$,
  the \textsc{Soundness Error Reduction with Guess~$k$} associated with ${(V_x, p)}$
  results in acceptance with probability at least~$\frac{1}{2}$
  for some $k$ in ${\{1, \dotsc, 2^{\op{l}(\abs{x})}\}}$,
  while if $x$ is in $A_\no$,
  it results in acceptance with probability at most ${2^{- \op{p}(\abs{x})}}$
  for any $k$ in ${\{1, \dotsc, 2^{\op{l}(\abs{x})}\}}$.
  Hence, obviously from its construction,
  the \textsc{Soundness Error Reduction with Random Guess} associated with ${(V_x, p)}$
  results in acceptance with probability at least
  \[
    2^{- \op{l}(\abs{x})} \cdot \frac{1}{2}
    >
    \frac{\op{c}(\abs{x}) - \op{s}(\abs{x})}{2 \sqrt{6 \op{p}(\abs{x})}}
    \cdot
    \frac{1}{2}
    =
    \frac{\op{c}(\abs{x}) - \op{s}(\abs{x})}{4 \sqrt{6 \op{p}(\abs{x})}}
  \]
  if $x$ is in $A_\yes$,
  and at most ${2^{- \op{p}(\abs{x})}}$
  if $x$ is in $A_\no$.
  This shows the completeness and soundness.

  From the structures of
  the \textsc{Additive Adjustment Procedure},
  \textsc{Reflection Procedure},
  and the \textsc{AND-Type Repetition Procedure},
  the \textsc{Soundness Error Reduction with Guess~$k$} associated with ${(V_x, p)}$
  uses extra workspace
  (relative to $V_x$)
  of ${\op{\delta}_1 (\abs{x})}$~qubits
  for the function~$\function{\delta_1}{\Nonnegative}{\Natural}$ defined by
  ${\delta_1 = l + \ceil{\log (2N + 1)} + 1}$,
  where
  ${l = \bigceil{\frac{1}{2} \log \frac{6p}{(c - s)^2}}}$
  and
  ${N = \bigceil{\frac{p}{2 (c - s)^2}}}$.
  Hence, $\delta_1$ is clearly logarithmic with respect to $\frac{p}{c-s}$.
  Therefore,
  the \textsc{Soundness Error-Reduction with Random Guess} associated with ${(V_x, p)}$
  uses extra workspace
  (relative to $V_x$)
  of logarithmically many qubits with respect to $\frac{p}{c-s}$ also
  (which is determined by a function~${\delta = \delta_1 + 2l}$,
  as the random guess may be implemented by preparing a sufficiently many number of EPR pairs
  and using each half of them),
  as desired.
\end{proof}


\paragraph{Space-efficient amplification based on a random guess}

Again fix arbitrarily a function~$\function{p}{\Nonnegative}{\Natural}$
and functions~$\function{c,s}{\Nonnegative}{[0,1]}$ satisfying ${c > s}$.
Let $\function{q, N}{\Nonnegative}{\Natural}$ be functions defined by
\[
  q
  =
  \biggceil{2 \biggl( p + \log \frac{6p}{c - s} + 1 \biggr)},
  \quad
  N
  =
  \biggceil{\frac{2 \sqrt{6q}}{c - s} \cdot p}.
\]

Fix an input~$x$.
Given the pair~${(V_x, p)}$,
consider the \textsc{Soundness Error-Reduction with Random Guess} associated with ${(V_x, q)}$.
Let $V'_x$ be the unitary transformation induced by it,
let $\Pi'_\init$ be the projection onto the subspace spanned by the legal initial states of it,
and let $\Pi'_\acc$ be the projection onto the subspace spanned by the accepting states of it.
From the triplet~${\bigl( V'_x, \Pi'_\init, \Pi'_\acc \bigr)}$
and a positive integer~${\op{N}(\abs{x})}$,
one constructs the \textsc{OR-Type Repetition Procedure}
associated with ${\bigl( V'_x, \Pi'_\init, \Pi'_\acc, \op{N}(\abs{x}) \bigr)}$,
and performs it.
The resulting procedure is called the \textsc{Space-Efficient Amplification Based on Random Guess}
and is summarized in Figure~\ref{Figure: Space-Efficient Amplification Based on Random Guess}.

\begin{figure}[t!]
  \begin{algorithm*}
        {
          \textsc{Space-Efficient Amplification Based on Random Guess}
          associated with $\boldsymbol{(V_x, p)}$
        }
    \begin{flushright}
      \begin{minipage}{0.9727\textwidth}
        Define functions~$q$~and~$N$ by
        ${q = \bigceil{2 \bigl( p + \log \frac{6p}{c - s} + 1 \bigr)}}$
        and
        ${N = \Bigceil{\frac{2 \sqrt{6q}}{c - s} \cdot p}}$.
        Consider the \textsc{Soundness Error Reduction with Random Guess} associated with ${(V_x, q)}$.
        Let $V'_x$ be the unitary transformation induced by it,
        let $\Pi'_\init$ be the projection onto the subspace spanned by the legal initial states of it,
        and let $\Pi'_\acc$ be the projection onto the subspace spanned by the accepting states of it.\\
        Perform the \textsc{OR-Type Repetition Procedure}
        associated with ${\bigl( V'_x, \Pi'_\init, \Pi'_\acc, \op{N}(\abs{x}) \bigr)}$.
      \end{minipage}
    \end{flushright}
    \vspace{0.75\baselineskip}
  \end{algorithm*}
  \caption{
    The \textsc{Space-Efficient Amplification Based on Random Guess}.
  }
  \label{Figure: Space-Efficient Amplification Based on Random Guess}
\end{figure}

Now Theorem~\ref{Theorem: main theorem},
the main theorem of this paper,
is ready to be proved
by using the \textsc{Space-Efficient Amplification Based on Random Guess}
combined with the properties of the \textsc{Soundness Error Reduction with Random Guess}
used for proving Lemma~\ref{Lemma: soundness error reduction with a random guess}.

\begin{proof}
      [%
        Proof of Theorem~\ref{Theorem: main theorem}
        (via the exactly implementable construction based on a random guess)
      ]
  Let ${A = (A_\yes, A_\no)}$ be a problem in ${\unitaryQMASPACE[l_\regV, l_\regM](c, s)}$,
  and let ${V = \{ V_x \}_{x \in \Sigma^\ast}}$
  be the ${(l_\regV, l_\regM)}$-space-bounded quantum verifier witnessing this membership.
  Fix a function~$\function{p}{\Nonnegative}{\Natural}$ and an input~$x$ in $\Sigma^\ast$.
  Let 
  The theorem is proved by considering
  the \textsc{Space-Efficient Amplification Based on Random Guess} associated with ${(V_x, p)}$.

  Let $\function{q}{\Nonnegative}{\Natural}$ be the function defined by
  ${q = \bigceil{2 \bigl( p + \log \frac{6p}{c - s} + 1 \bigr)}}$.
  First consider the \textsc{Soundness Error Reduction with Random Guess}
  associated with ${(V_x, q)}$.
  Let $V'_x$ be the unitary transformation induced by it,
  let $\Pi'_\init$ be the projection onto the subspace spanned by the legal initial states of it,
  and let $\Pi'_\acc$ be the projection onto the subspace spanned by the accepting states of it.

  As the function~$q$ satisfies that ${q > 2 \log \frac{4 \sqrt{3}}{c - s}}$,
  and thus, that ${\frac{c - s}{4 \sqrt{6q}} > 2^{-q}}$,
  Lemma~\ref{Lemma: soundness error reduction with a random guess} and its proof ensure that
  $A$ is in
  ${\unitaryQMASPACE[l_\regV + \delta_1, l_\regM] \bigl( \frac{c - s}{4 \sqrt{6q}}, 2^{-q} \bigr)}$
  for some function~$\function{\delta_1}{\Nonnegative}{\Natural}$
  that is logarithmic with respect to ${\frac{q}{c - s}}$
  (and thus, with respect to ${\frac{p}{c - s}}$),
  and this inclusion is certified by
  the \textsc{Soundness Error Reduction with Random Guess} associated with ${(V_x, q)}$.
  This in particular implies that
  the Hermitian operator~%
  ${M'_x = \Pi'_\init \conjugate{(V'_x)} \Pi'_\acc V'_x \Pi'_\init}$
  has an eigenvalue at least~%
  ${\frac{\op{c}(\abs{x}) - \op{s}(\abs{x})}{4 \sqrt{6 \op{q}(\abs{x})}}}$
  if $x$ is in $A_\yes$,
  while all the eigenvalues of $M'_x$ are at most~$2^{- \op{q}(\abs{x})}$
  if $x$ is in $A_\no$.

  Now consider the \textsc{OR-Type Repetition Procedure}
  associated with ${\bigl( V'_x, \Pi'_\init, \Pi'_\acc, \op{N}(\abs{x}) \bigr)}$,
  which is exactly what
  the \textsc{Space-Efficient Error Reduction Based on Random Guess} associated with ${(V_x, p)}$
  performs.
  By Proposition~\ref{Proposition: properties of OR-Type Repetition Procedure},
  the \textsc{OR-Type Repetition Procedure}
  associated with ${\bigl( V'_x, \Pi'_\init, \Pi'_\acc, \op{N}(\abs{x}) \bigr)}$
  results in acceptance with probability at least
  \[
    \begin{split}
      \hspace{5mm}
      &
      \hspace{-5mm}
      1
      -
      \biggl(
        1 - \frac{\op{c}(\abs{x}) - \op{s}(\abs{x})}{4 \sqrt{6 \op{q}(\abs{x})}}
      \biggr)^{2 \op{N}(\abs{x})}
      \\
      &
      \geq
      1
      -
      \biggl(
        1 - \frac{\op{c}(\abs{x}) - \op{s}(\abs{x})}{4 \sqrt{6 \op{q}(\abs{x})}}
      \biggr)^{\frac{4 \sqrt{6 \op{q}(\abs{x})}}{\op{c}(\abs{x}) - \op{s}(\abs{x})} \cdot \op{p}(\abs{x})}
      >
      1 - e^{- \op{p}(\abs{x})}
      >
      1 - 2^{- \op{p}(\abs{x})}
    \end{split}
  \]
  if $x$ is in $A_\yes$,
  and at most
  \[
    \begin{split}
      1 - \Bigl( 1 - 2^{- \op{q}(\abs{x})} \Bigr)^{2 \op{N}(\abs{x})}
      &
      <
      2^{- \op{q}(\abs{x}) + 1} \cdot \op{N}(\abs{x})
      \\
      &
      <
      2^{- \op{p}(\abs{x}) - \log \frac{6 \op{p}(\abs{x})}{\op{c}(\abs{x}) - \op{s}(\abs{x})}}
      \cdot
      2^{- \frac{1}{2} \op{q}(\abs{x})}
      \cdot
      \Biggl(
        \frac{2 \sqrt{6 \op{q}(\abs{x})}}{\op{c}(\abs{x}) - \op{s}(\abs{x})} \cdot \op{p}(\abs{x}) + 1
      \Biggr)
      \\
      &
      <
      2^{- \op{p}(\abs{x})}
      \cdot
      \frac{\op{c}(\abs{x}) - \op{s}(\abs{x})}{6 \op{p}(\abs{x})}
      \cdot
      \frac{1}{\sqrt{\op{q}(\abs{x})}}
      \cdot
      \biggl(
        \frac{6 \sqrt{\op{q}(\abs{x})}}{\op{c}(\abs{x}) - \op{s}(\abs{x})} \cdot \op{p}(\abs{x})
      \biggr)
      \\
      &
      <
    2^{- \op{p}(\abs{x})}
    \end{split}
  \]
  if $x$ is in $A_\no$,
  where the third inequality uses the fact that ${2 \sqrt{6} + 1  < 6}$,
  and the completeness and soundness follows.

  The \textsc{OR-Type Repetition Procedure}
  associated with
  ${\bigl( V'_x, \Pi'_\init, \Pi'_\acc, \op{N}(\abs{x}) \bigr)}$
  uses extra workspace
  (relative to $V'_x$)
  of ${\op{\delta}_2 (\abs{x})}$~qubits
  for the function~$\function{\delta_2}{\Nonnegative}{\Natural}$ defined by
  ${\delta_2 = \ceil{\log (2N + 1)}}$.
  As ${N = \bigceil{\frac{2 \sqrt{6q}}{c - s}}}$
  and ${q = \bigceil{2 \bigl( p + \log \frac{6p}{c - s} + 1 \bigr)}}$,
  $\delta_2$ is clearly logarithmic with respect to $\frac{p}{c-s}$.
  Hence,
  the \textsc{Space-Efficient Error Reduction Based on Random Guess} associated with ${(V_x, p)}$
  uses extra workspace
  (relative to $V_x$)
  of logarithmically many qubits with respect to $\frac{p}{c-s}$ also
  (which is determined by a function~${\delta = \delta_1 + \delta_2}$),
  as desired.
\end{proof}


\subsection*{Acknowledgements}

BF and CYL are supported by the Department of Defense.
HK and HN are supported by
the Grant-in-Aid for Scientific Research~(A)~No.~24240001
of the Japan Society for the Promotion of Science.
TM is supported by
the Program to Disseminate Tenure Tracking System
of the Ministry of Education, Culture, Sports, Science and Technology in Japan,
the Grant-in-Aid for Scientific Research on Innovative Areas~No.~15H00850
of the Ministry of Education, Culture, Sports, Science and Technology in Japan,
and
the Grant-in-Aid for Young Scientists~(B)~No.~26730003
of the Japan Society for the Promotion of Science.
HN is also supported by
the Grant-in-Aid for Scientific Research on Innovative Areas~No.~24106009
of the Ministry of Education, Culture, Sports, Science and Technology in Japan,
which HK is also grateful to.
HN further acknowledges support from
the Grant-in-Aid for Scientific Research~(C)~No.~25330012
of the Japan Society for the Promotion of Science.



\providecommand{\noopsort}[1]{}

\end{document}